\documentclass[runningheads]{llncs}

\usepackage{amssymb}
\usepackage{amsmath}
\usepackage{graphicx}
\usepackage{bussproofs}
\usepackage{fancybox}
\usepackage{cmll}
\usepackage{stmaryrd}
\usepackage{multirow}
\usepackage{bbm}
\usepackage{bm}
\usepackage{dsfont}

\usepackage{hyperref}
\usepackage{proof}
\usepackage{hhline}
\usepackage{xspace}
\usepackage{booktabs}
\usepackage{subcaption}
\usepackage{wrapfig}
\usepackage{array}
\usepackage{arydshln}
\usepackage{commath}
\usepackage{complexity}
\usepackage{microtype}
\usepackage{xcolor}
\usepackage[all,cmtip]{xy}

\newcommand{\tm}{t}
\newcommand{\tmtwo}{u}
\newcommand{\tmthree}{r}

\newcommand{\var}{x}
\newcommand{\vartwo}{y}
\newcommand{\varthree}{z}

\newcommand{\la}[2]{\l#1.#2}

\newcommand{\val}{v}
\newcommand{\valtwo}{w}

\newcommand{\ctxholep}[1]{[#1]}
\newcommand{\ctxhole}{\ctxholep{\cdot}}

\newcommand{\evctx}{E}
\newcommand{\Evctxs}{\mathbb{E}}

\newcommand{\nbvctxtwo}[1]{\nbvctxtwo{#1}}

\newcommand{\defeq}{:=}

\newcommand{\grameq}{::=}

\newcommand{\isub}[2]{\{#1/#2\}}

\newcommand{\grammarpipe}{\mathrel{\big |}}

\renewcommand{\l}{\lambda}
\newcommand{\ie}{\textit{i.e.}\xspace}
\newcommand{\eg}{\textit{e.g.}\xspace}
\newcommand{\ih}{\textit{i.h.}\xspace}

\newenvironment{varitemize}
{
\begin{list}{\labelitemi}
{\setlength{\itemsep}{0pt}
 \setlength{\topsep}{0pt}
 \setlength{\parsep}{0pt}
 \setlength{\partopsep}{0pt}
 \setlength{\leftmargin}{15pt}
 \setlength{\rightmargin}{0pt}
 \setlength{\itemindent}{0pt}
 \setlength{\labelsep}{5pt}
 \setlength{\labelwidth}{10pt}
}}
{
 \end{list} 
}

\newcounter{numberone}

\newenvironment{varenumerate}
{
\begin{list}{\arabic{numberone}.}
{
  \usecounter{numberone}
  \setlength{\itemsep}{0pt}
  \setlength{\topsep}{0pt}
  \setlength{\parsep}{0pt}
  \setlength{\partopsep}{0pt}
  \setlength{\leftmargin}{15pt}
  \setlength{\rightmargin}{0pt}
  \setlength{\itemindent}{0pt}
  \setlength{\labelsep}{5pt}
  \setlength{\labelwidth}{15pt}
}}
{
\end{list} 
}

\newcounter{numbertwo}

\newcommand{\linty}{A}
\newcommand{\lintytwo}{B}

\newcommand{\arr}[2]{#1\rightarrow #2}

\newcommand{\tye}{\Gamma}
\newcommand{\tyetwo}{\Delta}

\newcommand{\tjudg}[3]{#1\vdash #2:#3}
\newcommand{\ltjudg}[3]{#1\kleisli\vdash #2:#3}

\newcommand{\tyd}{\pi}
\newcommand{\tydtwo}{\tyd'}
\newcommand{\tydthree}{\tyd''}
\newcommand{\pof}{\;\triangleright\,}

\renewcommand{\int}[1]{\{#1\}}

\newcommand{\intty}{I}
\newcommand{\inttytwo}{J}
\newcommand{\monty}{M}
\newcommand{\montytwo}{N}
\newcommand{\mfunc}{T}

\newcommand{\op}[1]{\mathtt{op}(#1)}
\newcommand{\opfun}{\mathtt{op}}

\newcommand{\unit}{\eta}
\newcommand{\multiplication}{\mu}
\newcommand{\Tys}{\mathbb{\linty}}
\newcommand{\IntTys}{\mathbb{\intty}}
\newcommand{\MonTys}{\mathbb{\monty}}

\newcommand{\Set}{\mathbf{Set}}
\newcommand{\rel}{\mathbf{Rel}}

\newcommand{\Comps}{\mathbb{C}}
\newcommand{\Vals}{\mathbb{V}}

\newcommand{\mapstoo}{\rightarrowtail}

\newcommand{\mapstob}{\rightarrowtail_\beta}
\newcommand{\mapstoeta}{\rightarrowtail_\eta}

\newcommand{\mapstoe}{\rightarrowtail_{\mathtt{e}}}

\newcommand{\gef}[1]{g(#1)}

\newcommand{\gefop}[1]{g_\mathtt{op}(#1)}
\newcommand{\gefopfun}{g_\mathtt{op}}

\newcommand{\rax}{\textsc{var}}
\newcommand{\rabs}{\textsc{abs}}
\newcommand{\rint}{\textsc{int}}
\newcommand{\runit}{\textsc{unit}}
\newcommand{\rextg}{\textsc{ext-g}}
\newcommand{\rextunit}{\textsc{ext-unit}}
\newcommand{\rbot}{\textsc{bot}}

\newcommand{\rop}{\textsc{op}}
\newcommand{\rapp}{\textsc{app}}

\newcommand{\obs}[1]{\mathsf{obs}(#1)}
\newcommand{\obsfun}{\mathsf{obs}}
\newcommand{\obsn}[2]{\mathsf{obs}^{#2}(#1)}

\newcommand{\supp}[1]{\mathsf{supp}(#1)}
\newcommand{\bindsymbol}{\scalebox{0.5}[1]{$>\!>=$}}
\newcommand{\bind}{\mathrel{\bindsymbol}}
\newcommand{\etm}{e}
\newcommand{\etmtwo}{\etm'}
\newcommand{\etmthree}{\etm''}

\newcommand{\out}[2]{\mathtt{out}_#1(#2)}

\newcommand{\closed}[1]{#1^{\bullet}}
\newcommand{\trel}{{\vdash}}
\newcommand{\trelone}{{\vdash}}

\newcommand{\tosem}{{\mapsto}}

\newcommand{\monad}{T}
\newcommand{\kleislirel}[1]{#1^{\dagger}}
\newcommand{\kleisli}[1]{#1^{\dagger}}

\def\tobar{\mathrel{\mkern3mu  \vcenter{\hbox{$\scriptscriptstyle+$}}%
            \mkern-12mu{\to}}}
\newcommand{\torel}{\tobar}
\newcommand{\relator}{\Phi}
\newcommand{\relone}{R}
\newcommand{\reltwo}{S}

\newcommand{\mcomp}{\fatsemi}
\newcommand{\sem}[1]{\llbracket #1 \rrbracket}
\newcommand{\GTys}{\mathbb{G}}
\newcommand{\gty}{G}
\newcommand{\emi}{\mathbf{0}}
\newcommand{\dom}[1]{\mathsf{dom}(#1)}
\newcommand{\ems}{\varepsilon}
\newcommand{\id}{\mathtt{id}}

\newcommand{\bext}[1]{\widehat{#1}}
\newcommand{\bextt}[1]{\widehat{\monad}{(#1)}}

\newcommand{\trelsem}{\models}

\newcommand{\idrel}{\bm{1}}
\newcommand{\cpoleq}{\sqsubseteq}
\newcommand{\cpogeq}{\sqsupseteq}
\newcommand{\lub}{\bigsqcup}
\newcommand{\bexttinf}[1]{\widehat{T}_{\cpogeq}(#1)}

\newcommand{\sjudg}[2]{\trelsem#1:#2}
\newcommand{\lsjudg}[2]{\hat{\trelsem}\,#1:#2}

\newcommand{\dcppo}{dcppo}
\newcommand{\join}{\bigsqcup}
\newcommand{\metalambda}{%
\mathop{%
\rlap{$\lambda$}%
\mkern2mu
\raisebox{.275ex}{$\lambda$}%
}%
}

\newcommand{\ccode}[1]{\mathtt{#1}}
\newcommand{\powerset}{\mathcal{P}}

\newcommand{\nepowerset}{\powerset^{+}}
\newcommand{\FPowerset}{\mathbb{P}_{\mathit{f}}}
\newcommand{\fpowerset}{\mathcal{P}_{\mathit{f}}}
\newcommand{\signature}{\Sigma}
\newcommand{\opone}{\sigma}
\newcommand{\optwo}{\rho}

\newcommand{\writer}{\mathcal{W}}
\newcommand{\maybe}{\mathcal{M}}

\newcommand{\tree}[1]{T_{#1}}
\newcommand{\eq}{\sim}

\newcommand{\dual}[1]{#1^{\circ}}
\newcommand{\hh}{\hdots}

\newcommand{\tone}{t}

\usepackage{mathbbol}

\begin{document}

\title{Monadic Intersection Types, Relationally}

\author{Francesco Gavazzo\inst{1}
\and
Riccardo Treglia\inst{2}
\and Gabriele Vanoni\inst{3}
}

\institute{Università di Padova \email{francesco.gavazzo@unipd.it}\and 
 King's College London, London, UK \email{riccardo.treglia@kcl.ac.uk} \and
IRIF, CNRS, Université Paris Cité \email{gabriele.vanoni@irif.fr}}

\maketitle

\begin{abstract}
We extend intersection types to a computational $\lambda$-calculus with algebraic operations \emph{à la} Plotkin and Power. We achieve this by considering monadic intersections---whereby computational effects appear not only in the operational semantics, 
but also in the \emph{type system}. Since in the effectful setting termination is not anymore the only property of interest, we want to analyze the interactive behavior of typed programs with the environment. Indeed, our type system is able to characterize the natural notion of observation, both in the finite and in the infinitary setting, and for a wide class of effects, such as output, cost, pure and probabilistic nondeterminism, and combinations thereof. The main technical tool is a novel combination of syntactic techniques with abstract relational reasoning, 
which allows us to lift all the required notions, e.g. of typability and logical relation, to the monadic setting.
\end{abstract}

\section{Introduction}
Type systems are a key aspect of programming languages, ensuring good behavior during the execution of programs, such as absence of errors, termination, or properties such as productivity,  safety, and reachability. Additionally, they ensure it in a \emph{compositional} way, that is, if programs are assembled according to the underlying type discipline, then the good behavior is ensured also for the composed program.

\subsubsection*{Intersection Types.} Type systems have solid roots in logic and proof theory, as witnessed by the Curry-Howard correspondence between simple types and intuitionistic natural deduction. However, in the theory of the $\l$-calculus, there is another use of types that has been studied at length: \emph{intersection types}. They were introduced by Coppo and Dezani{-}Ciancaglini in the late 70's \cite{DBLP:journals/aml/CoppoD78} to overcome the limitations of Curry's type discipline and enlarge the class of terms that can be typed. This is reached by means of a new type constructor, the \textit{intersection}. In this way, one can assign a finite set of types to a term, thus providing a form of finite polymorphism.

Similarly to simple types, intersection types ensure termination. In contrast to most notions of types, however, they also \emph{characterize} termination, that is, they type \emph{all} terminating $\l$-terms. They can be seen as a compositional way of defining operational semantics, or, in a dual way, as a syntactic presentation of denotational models. Additionally, intersection types have shown to be remarkably flexible, since different termination forms can be characterized by tuning details of the type system (\eg, weak/strong full normalization, head/weak/call-by-value evaluation). Termination being only \emph{semi-decidable}, type inference cannot be decidable, which is why standard intersection types are somewhat incompatible with programming practice, although some restricted forms of intersection types have found applications in programming, see for example \cite{DBLP:conf/pldi/FreemanP91,DBLP:journals/mscs/Pierce97,DBLP:journals/jacm/FrischCB08,DBLP:conf/birthday/Dezani-Ciancaglini18}, or the recent survey by Bono and Dezani{-}Ciancaglini~\cite{DBLP:conf/lics/BonoD20}.


\subsubsection*{Beyond the Pure $\l$-Calculus.} 




Intersection types have been mostly developed in the realm of the pure $\l$-calculus. 
However, current programming languages are deeply effectful, exhibiting several simultaneous impure behaviours, such as raising exceptions, performing input/output operations, sampling from distributions, etc. Reasoning about effectful programs becomes a challenging goal since their behaviour becomes highly interactive, depending on the external environment. Type-based techniques seem very interesting in this respect, since they 
enable \emph{modular} and \emph{compositional} analysis of program behaviour. In particular, intersection type systems have already been successfully adapted to some concrete computational effects, such as probabilistic~\cite{DBLP:conf/ppdp/BreuvartL18,DBLP:journals/pacmpl/LagoFR21}, and pure nondeterminism~\cite{DBLP:conf/fossacs/Tsukada014}.
In spite of the remarkable results achieved by each of these formalisms -- 
for instance, probabilistic intersection types have been proved to 
characterize almost-sure termination -- all of these come with a major drawback: 
they are tailored to the specific family of effects considered. 
This results in a lack of robustness and modularity when it comes 
to extending languages with new effects. For instance, probabilistic intersection 
types as they are, cannot cope with, \eg, languages with \emph{both} randomness and 
output. 
This problem can be fixed (output is a well-behaved effect that nicely interacts with 
probabilistic nondeterminism) but in highly non-modular way. 
One has, in fact, to re-engineer the whole theoretical framework  behind probabilistic 
intersection types to account for output, too.


Now, the leading question is: \emph{can intersection types be scaled up in the case of effectful $\lambda$-calculi, in a modular way?}
In this paper, we answer this question in the affirmative by developing a 
general \emph{monadic intersection type system} for an untyped 
computational $\lambda$-calculus~ \cite{Moggi'89,Moggi'91} with algebraic operations 
\emph{\`a la} Plotkin and Power \cite{PlotkinP02}. In fact, our system covers both finitary and infinitary effectful operational behaviours in a sound and complete way, and generalises 
existing effectful intersection type systems, such as 
probabilistic intersection types. To achieve this result, we combine state-of-the-art techniques in 
monadic semantics, intersection types, and relational reasoning, in a novel and nontrivial way.

\subsubsection*{Monadic Semantics.}
As we have already mentioned, our work is grounded on the theoretical framework pioneered by Moggi \cite{Moggi'89,Moggi'91}, which describes computational effects via monads~\cite{MacLane/Book/1971}.
Moggi's work described how monads could give denotational semantics to effectful programs,
but did not tell anything
about how computational effects are actually produced. Plotkin and Power~\cite{PlotkinP02,PP03} introduced {\em algebraic effects} as a way of giving monadic semantics, in the style of Moggi, to certain operations which actually produce computational side effects.
The core syntax of effectful languages can be thus given in terms of computational 
calculi with effect-triggering operations. But what about the operational semantics and, most importantly, 
the (intersection) type system? 

There is a well-known way to give effectful 
operational semantics to programming languages, namely
making operational semantics effectful itself. That is, if we model 
the operational semantics of a language using a transition relation between terms, 
then we can make such a relation monadic  
by relating terms with \emph{monadic terms},
i.e. terms encapsulated by a monad, encoding the effects produced 
during the computation. 
Such relations---i.e. relations of the form $\relone \subseteq A \times T(B)$, 
with $T$ a monad---are known as \emph{monadic} or \emph{Kleisli relations} 
and have been successfully used to give operational semantics to 
monadic calculi~\cite{DBLP:conf/lics/GavazzoF21}.

What about the type system? In its bare essence, it is 
given by a relation between terms and types, hence leading to a situation 
similar to the one of operational semantics. The analogy is no coincidence: as 
we obtain monadic operational semantics relying on the theory 
of monadic relations, the very same theory allows us to 
define monadic type systems: a monadic typing relation associates 
terms with \emph{monadic types}. 
This idea has already been exploited by Dal Lago and collaborators~\cite{DBLP:conf/ppdp/BreuvartL18,DBLP:journals/pacmpl/LagoFR21},
who realised that to extend intersection types to probabilistic languages 
one has to type expressions not with types, but with 
\emph{distribution of types}.\footnote{Actually, as we will show in Section~\ref{sec:mon-types}, distributions do not behave well in this case. This is why in \cite{DBLP:conf/ppdp/BreuvartL18} convex sets of distributions are used. Multidistributions are instead used in \cite{DBLP:journals/pacmpl/LagoFR21}. } This is nothing but a monadic typing relation 
instantiated to the distribution monad.

\subsubsection*{Relational Reasoning.} Working at the abstract level of monadic relations gives several advantages, 
both in terms of modularity and expressiveness. 
Concerning the former, we shall develop abstract proof techniques 
that allow us to give proofs of subject reduction and expansion 
independently of the underlying monad. Such factorization results rely 
on the theory of (monadic) relational extensions.  Here, one studies 
how to extend a monadic relation $\relone \subseteq A \times T(B)$, such as the one modeling one-step operational semantics, or typing between terms and monadic types, 
to a relation $\kleisli{\relone} \subseteq T(A) \times T(B)$, necessary to model operational semantics, or typing, of \emph{monadic} terms. 
Such extensions, even if canonical, do not exist in general, 
and a celebrated result by Barr~\cite{Barr/LMM/1970} gives necessary and sufficient 
conditions for the existence of relational extensions: monads must be 
\emph{weakly cartesian} \cite{wc-1,wc-2}. Intuitively, being weakly cartesian means that during the evaluation there is no loss of information about the performed effects. This is a form of reversibility, that is needed, \eg, in subject expansion, 

This restriction rules out from our analysis monads such as the powerset or the distribution monad. 
Although that may appear as a weakness of our framework, it actually exploits a 
nontrivial limitative result that has already been observed in different forms in the 
literature: results involving forms of reversibility, such as 
subject expansion, are simply not available when monads are not weakly cartesian. 
This is the deep reason why probabilistic intersection types are defined via
 the multi-distribution, rather than distribution, monad. 
Remarkably, the same kind of limitative result has also been proved in the setting 
of monadic operational semantics and rewriting~\cite{DBLP:conf/lics/GavazzoF21}.

\subsubsection*{Contributions.} In this paper, we introduce the first (to the best of our knowledge) intersection type system handling the computational $\l$-calculus with algebraic operations. This is done by letting not only terms, but also intersection types be monadic. The main idea of intersection types is that they are a static way to describe the mechanism of evaluation of programs. Since the operational semantics of effectful languages can be conveniently described by evaluation inside monads, it is natural to embed also intersection types inside them. This way, we are able to push forward the correspondence between intersection type derivations and term evaluation to the effectful/monadic setting. More precisely, we develop several contributions and theoretical advances:
\begin{varitemize}
    \item \emph{The Type System:} We provide the first idempotent\footnote{The choice of the idempotent variant of intersection types should not be taken too strictly. All the results of the paper hold also turning intersections into multisets, \emph{mutatis mutandis}. Moreover, the meta-theory in the idempotent case is more involved (requiring logical relations to prove soundness), and we show this way the strength of our approach. Still, it is not an exercise of style, because intersection type systems used to formalize higher-order model checking algorithms (see Sec.~\ref{sec:conclusion} for a more detailed discussion) \emph{must} be idempotent, since otherwise one would lose decidability.} intersection type system for a $\l$-calculus with algebraic operations which is parametric in the underlying monad. We design the type system in such a way that not only terms, but also types become monadic.
    \item \emph{Characterization of Observable Behavior:} Differently from the pure untyped setting, in which intersection types characterize (different forms of) termination of programs, in the effectful setting we would like to characterize via the type system \emph{all} the effects produced during the evaluation. Indeed, we obtain such a result by generalizing standard soundness and completeness theorems, via abstract relational techniques. In particular, observable behaviors of typable (\ie all the terminating) terms can be read out of their types.
    \item \emph{Intrinsic Limits:} Our approach comes with the already described intrinsic limits about the class of well-behaving monads (the weakly cartesian ones). Moreover, if one sticks with the finitary case, where the natural notion of convergence is must termination, another restriction on the kinds of admissible operations is needed. Indeed, also operations that erase arguments break the subject expansion, and thus the completeness of the system. Still, this restriction can be removed considering an infinitary semantics.
    \item \emph{The Infinitary Case:} Some interesting notions of observation, such as the probability of convergence in probabilistic calculi, are naturally infinitary. For this reason, we extend our type system to capture infinitary behaviors. Interestingly, we need to add just one typing rule to the previous (finitary) system, namely the one that can type every term with the bottom of the monad. Naturally, this extension requires the monads to satisfy more conditions (mainly domain theoretic ones). Remarkably, this way we are actually able to relax the constraint on non-erasing operations introduced in the finitary system. 
 \end{varitemize}   

\subsubsection*{Related Work.} To the best of our knowledge, this is the first work about monadic intersection type systems for effectful calculi with algebraic operations, which are parametric on the underlying monad. On an orthogonal axis, an intersection type system for a variant over Moggi's computational $\l$-calculus, but without any reference to algebraic operations, has been proposed in~\cite{deLiguoroTreglia20}, while intersections types have developed for calculi with continuations in~\cite{DBLP:journals/lmcs/BakelBd18} and, paired with union types, in~\cite{DBLP:journals/lmcs/KesnerV19}. 
With concrete monads, instead, various proposals have appeared for the state monad~\cite{DR07,deLiguoroT21a,deLiguoroT21b,DBLP:conf/wollic/AlvesKR23}, and the distribution monad~\cite{EhrPagTas14,DBLP:conf/ppdp/BreuvartL18,DBLP:journals/pacmpl/LagoFR21}.
Moreover, lifting the monad to the type system has already been done in a series of works by Dal Lago and coauthors, \eg to analyze complexity~\cite{ADLG'19,DLG'19} and recently for the state monad in~\cite{DBLP:conf/wollic/AlvesKR23}.
More on the programming side, intersection types have been proposed for a
 $\lambda$-calculus with computational  side effects and reference types in~\cite{Davies-Pfenning'00}, but without any reference to monads.


\subsubsection*{Proofs.} Omitted proofs are in the Appendix.

\section{Intersection Types and the CbV $\lambda$-Calculus}\label{sec:cbv-types}
 We devote this section to a gentle introduction to intersection type systems. For the moment, we do not consider effectful calculi and we set our analysis in the (almost) standard setting of Plotkin's call-by-value (CbV) $\lambda$-calculus~\cite{DBLP:journals/tcs/Plotkin75}. Actually, the calculus we present is the kernel of CbV $\lambda$-calculus, that is as expressive as the Plotkin's~\cite{Acc15,FGdT23}, but allows only a restricted form of term application.
\subsubsection*{The (kernel) CbV $\lambda$-Calculus.} Given a countable set of variables $\mathcal{V}$, values and computations are defined by mutual induction as follows:
\begin{center}$\begin{array}{rlcl}
	\textsc{Computations} & \Comps\ni\tm,\tmtwo & \grameq & \val 
	\grammarpipe 
	\val\tm\\[3pt]
	\textsc{Values} & \Vals\ni\val,\valtwo & \grameq & \var\in\mathcal{V} \grammarpipe
	\la\var\tm\\[3pt]
		\textsc{Eval. Contexts}  &\Evctxs\ni\evctx & \grameq & \ctxhole 
	\grammarpipe \val\evctx 
\end{array}$\end{center}
\emph{Free} and \emph{bound variables} are defined as usual: $\la\var\tm$ binds $\var$ in $\tm$. Terms are considered modulo $\alpha$-equivalence. Capture-avoiding (meta-level) substitution of $\tmtwo$ for all the free occurrences of $\var$ in $\tm$ is written $\tm\isub\var\tmtwo$. As it is customary when working in the CbV setting, we restrict ourselves to \emph{closed} terms, \ie, we consider only terms without free variables. This means that the normal forms are all and only the closed values, noted $\closed\Vals$, \ie the $\lambda$-abstractions. 
The traditional $\beta$ rule is restricted to values, \ie only (closed) values can be substituted:
$
(\la\var\tm)\val\rightarrowtail\tm\isub\var\val
$.
The deterministic operational semantics $\mapsto$ is obtained by closing the $\beta$ rule (by value) $\rightarrowtail$ w.r.t. all evaluation contexts.
Please notice that although we restricted term application to have a value as the left subterm, we can recover the usual application 
between terms as $\tm\tmtwo \defeq 
(\la\var\var\tmtwo)\tm$, where $\var$ is a fresh variable.

\begin{figure}[t]
    \[
\begin{array}{c@{\hspace{0.8cm}}c}
	\infer[\rax]{\tjudg{\tye,\var:\intty}\var\linty}{\linty\in\intty} & 
	 \infer[\rint]{ \tjudg{\tye}{\val}{\int{\linty_i}_{i\in 
	F}} }{\left[ \tjudg{\tye}{\val}{\linty_i}\right]_{i\in F}}\\[7pt]
	
	 \infer[\rabs]{\tjudg{\tye}{\la\var\tm}{\arr\intty\inttytwo}} 
	 {\tjudg{\tye,\var:\intty}{\tm}{\inttytwo}}
	  & 
	
 	\infer[\rapp]{\tjudg{\tye}{\val\tm}{\inttytwo}}
 	{  \tjudg{\tye}{\val}{\arr{\intty}{\inttytwo}}
 	&  \tjudg{\tye}{\tm}{\intty}}
\end{array}
\]
\vspace{-8pt}
\caption{The intersection type system for closed call-by-value.}
\vspace{-8pt}
\label{fig:cbv-system}
\end{figure}

\subsubsection*{Intersection Types for the CbV $\lambda$-Calculus.}The CbV $\lambda$-Calculus is a universal, \ie Turing complete, model of computation. This way its halting problem is obviously undecidable. Nonetheless, terminating terms can be characterized by syntactic means. Intersection types are one way of doing this, in a compositional and logical way. The grammar of types is reminiscent of 
the call-by-value translation $(\cdot)^\mathtt{v}$ of intuitionistic logic into linear 
logic~\cite{girard_linear_1987,Ehrhard12} $(A\to B)^\mathtt{v}=\,!(A^\mathtt{v})\multimap\, !(B^\mathtt{v})$. Semantically, they can 
be seen as a syntactical presentation of filter models of the $\lambda$-calculus \cite{Coppo-et.al'84}. The grammar for types is based on \emph{two} layers of types, defined in a
mutually recursive way, \emph{value} types $\linty$, and \emph{intersections} (\ie sets) $\intty$ of value types.
\[\begin{array}{rr@{\hspace{0.1cm}}c@{\hspace{0.1cm}}lcl}
	\textsc{Value Types} & \Tys&\ni&\linty & \grameq & \arr{\intty}{\inttytwo} \\[3pt]
	\textsc{Intersections} & \IntTys&\ni&\intty,\inttytwo & \grameq & 
	\int{\linty_1,\ldots,\linty_n}\quad n\geq 0\\[3pt]
 \textsc{Types} & \GTys&\ni&\gty&\grameq&\linty\grammarpipe\intty
\end{array}\]

\begin{remark}
Please notice that intersections can be empty. The empty intersection type $\emi\defeq\{\}$ stands for the type of \emph{erasable} terms, which in Closed CbV are just those terms evaluating to closed values (\ie $\lambda$-abstractions). In CbV, terminating terms and erasable terms coincide, as the argument of a $\beta$-reduction has to be evaluated before being erased (and so its evaluation has to terminate).
\end{remark}

Type environments, ranged over by $\tye,\tyetwo$, are total maps
from variables to intersection types such that only finitely
many variables are mapped to non-empty intersection types, and we write $\tye = 
\var_1:\intty_1,\ldots,\var_n:\intty_n$ if 
$\dom\tye = \set{\var_1,\ldots,\var_n}$. Type 
judgments have the form $\tjudg{\tye}{\tm}{\gty}$. The 
typing rules are in Fig.~\ref{fig:cbv-system}, where $F$ stands for a finite, possibly empty, set of indexes; type derivations are written $\tyd$ 
and we write 
$\tyd\pof\tjudg{\tye}{\tm}{\gty}$ for a type derivation $\tyd$ with the 
judgment $\tjudg{\tye}{\tm}{\gty}$ as its conclusion.

Intuitively, intersections are needed because, during the evaluation of a term $\tm$, a subterm of $\tm$ can assume different types. For example in $(\la\var\var\var)(\la\vartwo\vartwo)$, the argument $\la\vartwo\vartwo$ has type $\emi\to\emi$, when it substitutes the first occurrence of $\var$ in functional position, and has type $\emi$ when it substitutes the second occurrence of $\var$ in argument position. These different uses, which require different types, are encoded into the intersection type. Moreover, the type system, contrarily to what happens in call-by-name, and consistently with the rationale of CbV, needs arguments of applications to be typed (with an intersection) just once.
\begin{example}\label{ex:pure}
We provide the type derivation for the term $\tjudg{}{(\la\var\var\var)(\mathsf{II})}{\emi}$, where $\mathsf{I}\defeq\la\var\var$, in Fig.~\ref{fig:derivation}. One can notice that our example term, being CbV-normalizing, can be typed with $\emi$.

\begin{figure*}[t]
{\scriptsize\[
  \infer[\rapp]
  {\tjudg{}{(\la\var\var\var)(\mathsf{II})}{\emi}}
  { \infer[\rabs]
    {\tjudg{}{\la\var{\var\var}}{\int{\id}\to\emi}}
    {\infer[\rapp]{\tjudg{\var:\int{\id}}{\var\var}{\emi}}
    {\infer[\rax]{\tjudg{\var:\int{\id}}\var{\id}}{}
    &\infer[\rint]{{\tjudg{\var:\int{\id}}\var\emi}}{}}}
  & \infer[\rapp]   
    {\tjudg{}{(\la\vartwo{\vartwo})(\la\varthree{\varthree})}{\int{{\emi}\to\emi}}}     
    { \infer[\rabs]
    {\tjudg{}{\la\vartwo{\vartwo}}{\int{\id}\to \int{\id}}}
    {\infer[\rint]
    {\tjudg{\vartwo: \int{\id}}{\vartwo}{\int{\id}}}
    {\infer[\rax]{\tjudg{\vartwo: \int{\id}}{\vartwo}{\id}}{}}} & 
    \infer[\rint]{\tjudg{}{\la\varthree{\varthree}}{\int{\id}}}
    {    \infer[\rabs]{\tjudg{}{\la\varthree{\varthree}}{\id}}{    
    \infer[\rint]{\tjudg{\varthree:\emi}\varthree\emi}{}}
    }}
  }
\]}
    \vspace{-18pt}
    \caption{Type derivation for $\tjudg{}{(\la\var\var\var)(\mathsf{II})}{\emi}$. We set $\id\defeq\emi\to\emi$.}
    \vspace{-12pt}
    \label{fig:derivation}
\end{figure*}

\end{example}
\subsubsection*{Characterization of Termination.} Intersection types characterize Closed 
CbV termination, that is, they type all and only those $\lambda$-terms that terminate
with respect to Closed CbV. We give a very brief overview of how this result is achieved. In the following sections, we shall prove all these results in the effectful setting. The reader could, however, benefit from the exposition of the main steps in this simpler setting.

Similarly to more traditional type systems, this intersection type system enjoys subject reduction, \ie types are preserved under reduction.
\begin{lemma}[Subject Reduction]
    Let $\tm$ be a closed $\lambda$-term. If $\tjudg{}{\tm}{\intty}$ and $\tm\mapsto\tmtwo$, then 
	$\tjudg{}{\tmtwo}{\intty}$.
\end{lemma}
Moreover, as with simple types, all typable terms terminate.
\begin{proposition}[Soundness]
    Let $\tm$ be a closed $\lambda$-term. If $\tjudg{}{\tm}{\intty}$, then $\tm$ has normal form.
\end{proposition}
Contrarily to simple types, intersection types satisfy also subject \emph{expansion}. This means that types are preserved also by backward reductions (\ie expansions).
\begin{lemma}[Subject Expansion]
	Let $\tm$ be a closed $\lambda$-term. If $\tjudg{}{\tmtwo}{\intty}$ and 
	$\tm\mapsto\tmtwo$, then 
	$\tjudg{}{\tm}{\intty}$.
\end{lemma}
Together with the fact that normal forms, \ie $\lambda$-abstractions, can always be typed with the empty type $\emi$, this gives the completeness of the type system, \ie the fact that every terminating term is typable.
\begin{proposition}[Completeness]
    Let $\tm$ be a closed $\lambda$-term. If $\tm$ has normal form, then there exists an intersection type $\intty$ such that $\tjudg{}{\tm}{\intty}$.
\end{proposition}
Putting soundness and completeness together, we obtain the full characterization of termination via typability.
\begin{theorem}[Characterization]
    Let $\tm$ be a closed $\lambda$-term. Then there exists an intersection type $\intty$ such that $\tjudg{}{\tm}{\intty}$ if and only if $\tm$ has normal form.
\end{theorem}

\section{Preliminaries on Monads, Algebraic Effects, Operations}\label{sec:preliminaries}

In this section, we recall some preliminary notions 
on monads~\cite{MacLane/Book/1971}, algebras 
\cite{sankappanavar1981course}, and 
relational reasoning \cite{DBLP:books/cu/Schmidt2011},
 that will be central to the rest of this paper. 
Due to space constraints, there is no hope to be comprehensive, and 
thus we assume the reader to have minimal familiarity 
with those fields. 
Unless explicitly stated, we work in the category $\Set$ of sets and 
functions and we tacitly restrict all definitions to it. 
Since we will extensively work with relations, we employ the relational notation 
even for functions, writing $f;g: A \to C$ for the composition (in diagrammatic order) 
of $f: A \to B$ and $g: B \to C$, and $\idrel_A: A \to A$ (mostly omitting subscripts) 
for the identity function. 

\subsubsection*{Monads and Algebraic Effects.}
We use monads
 \cite{Moggi'89,Moggi'91} 
to model computational effects.

\begin{definition}[Monad]
\label{def:monad}
A \emph{monad} (on $\Set$) is a triple 
$(\monad, \unit, \mu)$ 
consisting of a functor $\monad$ (on $\Set$) 
together with two natural transformations: 
$\unit: 1_{\Set} \Rightarrow \monad$ (called \emph{unit}) and 
$\multiplication: \monad \monad \Rightarrow \monad$ (called \emph{multiplication}) 
subject to the following laws: 
$\eta; \mu = T(\eta); \mu = \idrel$ 
and $T(\mu); \mu =\mu;\mu$.
\end{definition}

Given a monad $(T, \eta, \mu)$ we oftentimes identify it with its carrier 
functor. Moreover, we write $\kleisli{f}: T(A) \to T(B)$ for the Kleisli extension of 
$f: A \to T(B)$, where $\kleisli{f} \defeq T(f);\mu$, and 
$\bindsymbol$ for the binding operator induced by $\kleisli{-}$. 
Such an operator maps a monadic element $t \in T(A)$ and a monadic function 
$f: A \to T(B)$ to the monadic element $t \bind f$ in $T(B)$ defined 
as $\kleisli{f}(t)$. It is well-known that using this construction a monad could be presented also as a Kleisli triple $(\monad, \unit, \kleisli{-})$, or with the bind operation instead of $\mu$, \ie as $(T, \eta, \bindsymbol)$~\cite{Wadler'95}.

To model how actual effects are produced, Plotkin and Power 
\cite{PP01,PP03}
introduced the notion of 
an \emph{algebraic operation}, which we shall use to make calculi 
truly effectful.

\begin{definition}[Algebraic Operation]
Given a monad $(\monad, \unit, \multiplication)$, 
an $n$-ary \emph{algebraic operation} is a 
natural transformation $\alpha: \monad^n \Rightarrow \monad$ 
respecting the monad multiplication. 
\end{definition}
From an operational perspective, algebraic operations describe those 
operations whose execution is independent of the context in which they are executed. 

\begin{example}[Concrete Monads and Operations]
\label{ex:monads}
\begin{varenumerate}
    \item[]
  \item
    Divergent computations are modelled by the 
    \emph{maybe} or \emph{partiality} monad $(\mathcal{E}, \unit, 
    \multiplication)$, where 
    $ \mathcal{E}(A) \defeq A + \{\bot\}$, 
    $\unit$ is the left injection $\iota_{\ell}$, and 
    $\multiplication: ((A + \{\bot\}) + \{\bot\}) 
    \to A + \{\bot\}$ sends $\iota_{\ell}(\iota_{\ell}(x))$ to 
    $\iota_{\ell}(x)$, and all the rest to $\bot$.
    Therefore, an element in $\maybe A$ is either an element $a \in A$ 
    (meaning that we have a terminating computation returning $a$), or 
     the element $\bot$ (meaning that the computation diverges).
    As non-termination is an intrinsic feature of complete programming languages, 
     we do not consider 
     explicit operations to produce divergence. 
 Nonetheless, notice that 
we might consider the constant $\bot$ as a zero-ary operation 
 producing divergence (linguistically, this essentially corresponds to adding an 
always diverging constant $\mathtt{diverge}$).
\item Replacing $\{\bot\}$ with a set of errors $\mathit{Err}$, 
  we obtain the exception monad. Exceptions are produced by means of 
  $0$-ary operations $\mathtt{raise}_{e}$ indexed by elements in $\mathit{Err}$.
\item Probabilistic computations are modelled by the (finitary) distribution 
  monad $(\mathcal{D}, \eta, \mu)$, 
  where $\mathcal{D}(X)$ is the set of distributions over $X$ with 
  \emph{finite} support (where the support of a distribution $d\in\mathcal{D}(X)$ is defined as $\supp{d}\defeq\{x\in X\text{ such that }d(x)>0\}$), $\unit(x)$ is the Dirac distribution on $x$, and 
  $\mu(\Phi)(x) \defeq \sum_{\phi} \Phi(\phi) \cdot \phi(x)$. 
  Finite distributions are produced using weighted sums, \ie 
  $[0,1]$-indexed binary operations $+_p$ defined thus: 
  $(\phi_1 +_p \phi_2)(x) \defeq p \cdot \phi_1(x) + (1-p) \cdot \phi_2(x)$. 
  In a similar fashion, one defines the finitary subdistribution monad $\mathcal{D}^{\leq 1}$ 
  and the countably supported (sub)distribution monad $\mathcal{D}_{\omega}^{(\leq 1)}$. 
  In the former case, we simply allow distribution to have weight smaller than $1$, whereas 
  in the latter case, we allow distributions to have a countable support (i.e. 
  the set of elements where the distribution is non-null must be countable).
\item
    Computations with output are modelled by
    the \emph{writer} or \emph{output} monad
    $(\writer, \unit, 
    \multiplication)$, where
    $\writer A = W \times A$ and  $(W, 1, \cdot)$ is a monoid.  
    Unit and multiplication are defined by 
    $\unit(x) = (1,x)$ 
    and $\multiplication(w, (v,x)) = (w\cdot v, x)$. 
    Taking the monoid of words, then we can think of $(w,x)$ as the result of 
    a program printing $w$ and returning $x$. 
    If, instead, we take the monoid $(\mathbb{N}, 0, +)$, then we obtain 
    the \emph{cost} or \emph{complexity} monad \cite{Sands/Improvement-theory/1998}, whereby 
    we read $(n,x)$ as the result of a computation that produces $x$ 
    with cost $n$. 
    We consider a $W$-indexed family of unary operations 
    $\ccode{out}_w$ defined by 
    $\ccode{out}_w(v,x) = (w\cdot v,x)$. These are indeed algebraic. 
    When dealing with cost, one usually considers the single operation $\ccode{out}_1$, 
    usually written $\ccode{tick}$ or simply $\checkmark$.
  \item We model nondeterminism using the 
     powerset monad $(\powerset, \unit, 
     \multiplication)$, where 
     $\unit(x) = \{x\}$ and 
     $\multiplication(\mathcal{U}) = \bigcup \mathcal{U}$.
     We generate nondeterminism using (binary) set-theoretic union, 
     which is indeed algebraic.
     The finitary powerset monad $\FPowerset$ is obtained from 
     $\powerset$ by taking as underlying functor the finite powerset 
     functor $\fpowerset$. Similarly, the non-empty powerset monad 
     $\nepowerset$ is obtained by the taking the non-empty powerset 
     functor $\nepowerset$.
\end{varenumerate}
\end{example}

Most monads seen in Example~\ref{ex:monads} have \emph{countable support} 
in the sense that whenever $t \in \monad(A)$ there exists a countable 
set $\supp{t} \subseteq A$ upon which $t$ is built. 
Such a set is called the \emph{support} of $t$ and generalises 
the notion of a support one has for distributions. 
In general, not all monads have support and thus 
we restrict our analysis to such monads.
First, we consider only monads that preserve injections: 
that is, if $\iota: X \hookrightarrow A$ is an injection, then so is
$T(\iota): T(X) \hookrightarrow T(A)$.  
We regard $T(\iota)$ as a monadic inclusion and write 
$t \in T(X)$ if there exists (a necessarily unique) $s \in T(X)$ such that 
$T(\iota)(s) = t$.
Notice that all monads of Example~\ref{ex:monads} preserves injections and 
that if a monad preserves weak pullbacks (a condition we shall exploit in 
Section~\ref{sec:mon-types}), then it preserves injections. 

\begin{definition}[Support]
\begin{varenumerate}
\item[] 
\item
Given an element $t \in T(A)$, the support of $t$, if it exists, is the
smallest subset $\iota: X \hookrightarrow A$ such that 
$t \in T(X)$. We denote such a
set by $\supp{t}$.
\item 
We say that a monad $T$ has \emph{countable} (resp. \emph{finitary}) \emph{support} 
if any $t \in T(A)$ has countable (resp. finite) support --- \ie $\supp{t}$ exists and it 
is countable (resp. finite) --- for any set $A$.
\end{varenumerate}
\end{definition}

\begin{example}
   All the monads in Example~\ref{ex:monads} are countable,
with the exception of the (non-finitary) powerset monads. Nonetheless, we
can regard them as countable by taking their countable restriction. 
Indeed, as we shall apply them to countable sets (of terms), 
such a restriction is by no means a constraint. 
The output monad, the maybe monad, the finitary (sub)distribution monad, and the finitary powerset monad 
all have finitary support. For example, let us take the set $\mathcal{D}(\mathbb{N})$ of the probability distributions on natural numbers. If $\mathcal{D}(\mathbb{N})\ni d\defeq\frac{1}{3}\cdot 5,\frac{2}{3}\cdot 7$, applying the definition of support stated above, we obtain that the smallest set $X$ such that $d\in\mathcal{D}(X)$, which is $X=\{5,7\}$. This set matches exactly the one obtained from the standard definition of support for probability distributions given in Example~\ref{ex:monads}.
\end{example}

\subsubsection*{Algebraic Theories.} 

Since effects are ultimately produced by 
algebraic operations, we oftentimes  
describe computational effects by means of \emph{algebraic theories}, 
\ie via a collection of operations 
and equations.  

Recall that a \emph{signature} $\signature$ is a family of sets $\{\signature_k\}_{k \in \mathbb{N}}$, 
    the elements $\opone, \optwo, \hh$ of each $\signature_k$ being called $k$-ary \emph{operations}. 
    The set $\tree{\signature}(X)$ of $\signature$-\emph{terms} 
    (just terms) over $X$ is the least such that (1)
    $x \in \tree{\signature}(X)$ for any $x \in X$, and (2)
    $\opone(\tone_1, \hh, \tone_n) \in \tree{\signature}(X)$, whenever  
    $\tone_1, \hh, \tone_n \in \tree{\signature}(X)$. 
The construction of $\signature$-terms defines a functor 
$\tree{\signature}$ which is part of a monad  
whose unit is given by the subset inclusion injection $\iota: X \hookrightarrow \tree{\signature}{X}$ 
and whose multiplication is given by term 
substitution. 

An \emph{algebraic} or \emph{equational theory} 
over $\tree{\signature}(X)$ 
is given by a relation
    $E \subseteq \tree{\signature}(X) \times \tree{\signature}(X)$ 
    of equations between such terms. 
For a theory $E$, we write 
$\eq_{E}$ (or just $\eq$)
for the least congruence relation on terms that is closed under term substitution and contains $E$. 
The \emph{free $E$-theory} over $X$ 
is the quotient of $\tree{\signature}(X)$ by $\eq_{E}$. 
This construction gives a functor which is part of a 
monad, called the \emph{free theory monad} of $E$. 

\begin{example}[Concrete Algebraic Theories]
  \begin{varenumerate}
    \item[]
    \item The theory of divergence has a single $0$-ary operation 
      and no equation. Its free theory monad gives the maybe monad.
    \item The theory of nondeterminism consists of a single binary 
      operation $\vee$ together with the usual join semilattice 
      equations \cite{AbramskyJung/DomainTheory/1994}.
      Its associated free theory monad gives the 
      finitary non-empty powerset monad. If we drop the 
      idempotency equation $x \vee x \eq x$, we obtain 
      the theory of multisets~\cite{multisets} and the associated 
      multiset monad. If we also drop commutativity, we obtain the theory of lists 
      and the associated list monad.
    \item The theory of probabilistic nondeterminism has 
      binary operations 
      $+_p$ indexed by rational numbers $0 \leq p \leq 1$ 
      subject to the usual axioms of a barycentric algebra 
      \cite{stone49}:
      \begin{align*}
        x +_p x &\eq x;
        & 
        x +_1 y &\eq x;
        &
        x +_p y &\eq y +_{1-p} x;
      \end{align*}
      \vspace{-0.8cm}
      \begin{align*}
        x +_p (y +_q z) &\eq (x +_{\frac{p}{p + (1-p)q}} y) +_{p + (1-p)q} z.
      \end{align*}
      The free theory monad of this theory gives the finitary distribution monad.
      The theory of multi-distribution or indexed valuations (and their 
      corresponding monads) \cite{DBLP:journals/scp/AvanziniLY20,DBLP:journals/mscs/VaraccaW06,varacca2003probability}, 
      is obtained by dropping the idempotency axiom $x +_p x \eq x$. 
    \item Fixing a monoid $W$, the theory of the writer monad 
      has a unary operation $\ccode{out}_{w}$ for each 
      $w \in W$, and equations 
      \begin{align*}
      \ccode{out}_1(x) &\eq x 
      &
      \ccode{out}_w(\ccode{out}_v(x)) &\eq \ccode{out}_{w\cdot v}(x).
      \end{align*}
    \end{varenumerate}
\end{example}

\subsubsection*{Relations.}
We will extensively work with relations. 
We denote by $\rel$ the category with sets as objects 
and binary relations as arrows. As it is customary, we use the notation 
$\relone: A \torel B$ for a relation $\relone \subseteq A \times B$, 
and write $\rel(A,B)$ for the collection of relations 
of type $A \torel B$. 
We tacitly regard each function $f$ as a relation via its graph 
and write $\idrel_A: A \torel A$ for the identity relation, the latter 
being the graph of the identity function. We furthermore denote by 
$\relone; \reltwo: A \torel C$ the composition of 
$\relone: A \torel B$ and $\reltwo: B \torel C$,
and by $\dual{\relone}: B \torel A$ 
the dual or transpose of $\relone:A \torel B$. 


\section{Monadic Intersection Types}\label{sec:mon-types}
In this section, we present the monadic extension of the intersection type system for CbV presented in Section~\ref{sec:cbv-types}.

\subsubsection*{Effectful CbV.}
The target calculus of the remaining part of this work is an effectful 
extension of the CbV $\lambda$-calculus previously introduced. 
We follow the methodology of algebraic effects \cite{PlotkinP02} and 
fix a signature $\Sigma$ of effect triggering operations, 
as seen in the previous section. The calculus 
$\Lambda_{\Sigma}^{\mathsf{cbv}}$ is obtained by extending the grammar 
of the (kernel) CbV $\lambda$-calculus as follows:
\[
\Comps\ni\tm,\tmtwo \grameq 
	\cdots \mid \op{t_1, \hdots, t_n}
\]
As before, we denote by $\closed{\Comps}$ and $\closed{\Vals}$ the 
collection of closed computations and values, respectively. 
Finally, we write $\closed{\mathbb{R}}$ for the subset of $\closed{\Comps}$ of redexes, 
\ie computations of the form $(\lambda x.t)v$ or $\op{t_1, \hdots, t_n}$.

We give an operational semantics to 
$\Lambda_{\Sigma}^{\mathsf{cbv}}$ in monadic style \cite{DBLP:conf/fossacs/PlotkinP01,DBLP:conf/lics/GavazzoF21}. 
Let $(T, \eta, \mu)$ be an arbitrary but fixed monad with 
countable support. We assume that to each $n$-ary operation 
$\opfun \in \Sigma$ it is associated a $n$-ary algebraic 
operation $\gefopfun$ on $T$.

We define a function 
$\mapsto:\closed{\Comps} \to T(\closed{\Comps})$ on closed terms 
that performs a single computation step (possibly performing effects) 
by first defining ground reduction and then closing the latter under 
evaluation contexts. 
To improve readability, we write 
$t \mathrel{\tosem} e$ in place of ${\tosem}(t) = e$ and we refer to 
elements in $T(\Comps)$ (resp. $T(\Vals)$) as monadic (or effectful) computations 
(resp. values).

\begin{definition}[Operational Semantics]
We define the function $\mapstoo: \closed{\mathbb{R}} \cup \closed{\Vals} \to T(\closed{\Comps})$ as follows:
\[
\begin{array}{rcl}
	(\la\var\tm)\val&\mapstoo&\unit(\tm\isub\var\val)\\[3pt]
	\op{\tm_1,\ldots,\tm_n}&\mapstoo&\gefop{\unit(\tm_1),\ldots,\unit(\tm_n)}\\[3pt]
	\val &\mapstoo&\unit(\val)
\end{array}
\]
The function $\tosem: \closed{\Comps} \to T(\closed{\Comps})$ is then defined as the contextual closure of $\mapstoo$, \ie
$
\evctx[r] \mathrel{\tosem} e\, \bindsymbol\, (\metalambda\! u. \unit(\evctx[u]))
$, where $r$ is a redex and $r \mathrel{\mapstoo} e$.
\end{definition}

In this last definition the symbol $\metalambda$ has to be intended as a meta-lambda notation, \ie by $\metalambda\! u. \unit(\evctx[u])$ we mean the function $h:\closed{\Comps} \to T(\closed{\Comps})$ such that $h(u)\defeq\unit(\evctx[u])$. In particular, in the definition of the contextual closure we exploit the algebricity of the effects, making them commute with evaluation contexts. Moreover, notice that $\tosem$ is indeed a function. Consequently, we can rely on its Kleisli 
extension $\kleisli{\tosem}$ to reduce monadic computations. 
We write $\tosem^n$ for the $n$-iteration of $\tosem$, where 
$\tosem^0 \defeq \eta$ and $\tosem^{n+1} \defeq \tosem; \kleisli{(\tosem^n)}$. 
In a similar fashion, we write $\tosem^*$ for 
$\bigcup_n \tosem^n$.
Please notice that since algebraic operations are finitary, all monadic computations 
that a computation $t$ can achieve in finite times have finite support, meaning 
that $t \mathrel{\tosem^*} e$ implies that $\supp{e}$ is finite. 

\begin{definition}[Finitary Convergence]
    We say that a closed computation $t$ converges if there 
    exists $e \in T(\closed{\Comps})$ such that 
    $t \mathrel{\tosem^*} e$ and $\supp{e} \subseteq \closed{\Vals}$.\footnote{
    To be formally precise, here we should say that $\supp{e}$ belongs to
    the image of $\closed{\Vals}$ into $\closed{\Comps}$. In order to maintain
    the work as readable as 
    possible, we will be sloppy and here (and in similar situations) 
    simply identify $\closed{\Vals}$ and its image in $\closed{\Comps}$.} 
    In that case, we write $\sem t = e$.
\end{definition}
Please notice that $\sem\cdot$ is a \emph{partial} function. Indeed, terms can diverge.

\subsubsection*{The Monadic Type System.} The main idea behind the development of the type system is that not only terms, but also intersection types become monadic. The natural design choice is to follow the informal CbV translation of intuitionistic logic into linear logic combined with Moggi's translation:
\[
\arr{\linty}{\lintytwo}\cong\; !\linty \multimap{\mfunc(!\lintytwo)}
\]
A third level of (monadic) types is then added to the grammar of types: 
\[\begin{array}{r@{\hspace{1cm}}r@{\hspace{0.1cm}}c@{\hspace{0.1cm}}lcl}
	\textsc{Value Types} & \Tys&\ni&\linty & \grameq & \arr{\intty}{\monty} \\[3pt]
	\textsc{Intersections} & \IntTys&\ni&\intty & \grameq & 
	\int{\linty_1,...,\linty_n}\ n\geq 0\\[3pt]
	\textsc{Monadic Types} &  \monty,\montytwo&\ni&\MonTys & \grameq & \mfunc(\IntTys)\\[3pt]
 \textsc{Types} & \GTys&\ni&\gty&\grameq&\linty\grammarpipe\intty \grammarpipe\monty
\end{array}\]

\begin{figure}[t]
{\scriptsize \[
\begin{array}{c@{\hspace{0.4cm}}c@{\hspace{0.8cm}}c}
	\infer[\rax]{\tjudg{\tye,\var:\intty}\var\linty}{\linty\in\intty} & 
	\infer[\rapp]{\tjudg{\tye}{\val\tm}{\montytwo\bind{(\{\intty_i\Rightarrow\monty_i\}_{1\leq i\leq n})}}}
 	{ \left[ \tjudg{\tye}{\val}{\arr{\intty_i}{\monty_i}}\right]_{1\leq i\leq n}
 	&  \tjudg{\tye}{\tm}{\montytwo} & \supp{\montytwo}\subseteq\{\intty_1 ,...,\intty_n\}}
	\\[10pt]

	 \infer[\rabs]{\tjudg{\tye}{\la\var\tm}{\arr\intty\monty}} 
	 {\tjudg{\tye,\var:\intty}{\tm}{\monty}}
	  & 
	\infer[\rop]{\tjudg{\tye}{\op{\tm_1,\ldots,\tm_n}}{\gefop{\monty_1,\ldots,\monty_n}}} 
	{\left[
	\tjudg{\tye}{\tm_i}{\monty_i}\right]_{1\leq i\leq n}}\\[10pt]

        \infer[\rint]{ \tjudg{\tye}{\val}{\int{\linty_i}_{i\in 
	F}} }{\left[ \tjudg{\tye}{\val}{\linty_i}\right]_{i\in F}} & 
        \infer[\runit]{\tjudg{\tye}{\val}{\unit(\intty)}}{\tjudg{\tye}{\val}{\intty}}
\end{array}
\]}
\vspace{-8pt}
\caption{The monadic intersection type system.}
\vspace{-12pt}
\label{fig:mon-system}
\end{figure}

We maintain all the notations already presented in Section~\ref{sec:cbv-types} for the CbV type system. The typing rules are in Fig.~\ref{fig:mon-system}.
While rules $\rabs,\rax$ and $\rint$ are almost unchanged, the other rules deserve some comments. The rule $\runit$ is needed to give a monadic type to values. Since values do not produce any effect, they are injected into the monad just with $\unit$. Rule $\rapp$ types applications $\val\tm$ with a monadic type. The important point is that the subterm $\val$ in function position has to be typed many times, one for each element in the support of the (monadic) type of the argument $\tm$. Please notice (i) that with the notation $\{\intty_i\Rightarrow\monty_i\}_{1\leq i\leq n}$ we mean the function that maps pointwise $I_i$ to $M_i$ for each $1\leq i\leq n$; and (ii) the $\bindsymbol$ operator at the level of types.  This should not come unexpectedly, since types are monadic, and thus are composed using monadic laws. In particular, the effects produced by the argument, encoded in its type, have to be composed with the effects that will be generated by the rest of the computation (see the example below for more intuitions). Finally, rule $\rop$ types operations with the monadic type built with the corresponding algebraic operation, applied to the (monadic) types of the arguments of the operation itself.
\begin{example}
We provide in Fig.~\ref{fig:derivation2} the derivation for $\tjudg{}{(\la\var\var(\out{b}{\var}))(\out{a}{\mathsf{II}})}{(ab,\emi)}$, which is the very same term as in Example~\ref{ex:pure}, decorated with output operations. As monoid of words, we consider the monoid $\Sigma^*$ freely generated from the alphabet $\Sigma\defeq\{a,b\}$. One can notice that the assigned type contains all the information about the symbols printed on the output buffer during the evaluation of the term. Please notice how types are composed in the last rule $\rapp$. The right-hand side is typed only once, since the support of the monadic type on the left-hand side is a singleton. Notice that the bind operator $\bindsymbol:W{\times} X\to(X\to W{\times} X) \to W{\times} X$ in the case of the output monad is defined as: $(a,x)\bindsymbol f\defeq (ab,y)$ if $f(x)=(b,y)$. Intuitively, the usual composition is done, but for the fact that the strings ($a$ and $b$ in this case) are concatenated.
    \begin{figure*}[t]
{\scriptsize\[
\begin{array}{c}
\infer[\runit]{\pi:\ \tjudg{}{\la\varthree{\varthree}}{\unit(\int{\id})}}
    {   \infer[\rint]{\tjudg{}{\la\varthree{\varthree}}{\int\id}}{ \infer[\rabs]{\tjudg{}{\la\varthree{\varthree}}{\id}}{ \infer[\runit]{\tjudg{\varthree:\emi}\varthree{\unit(\emi)}}{
    \infer[\rint]{\tjudg{\varthree:\emi}\varthree\emi}{}}}
    }}\\
  \infer[\rapp]
  {\tjudg{}{(\la\var\var(\out{b}{\var}))(\out{a}{\mathsf{II}})}{(ab,\emi)}}
  { \infer[\rabs]
    {\tjudg{}{\la\var\var(\out{b}{\var})}{\int{\id}\to(b,\emi)}}
    {\infer[\rapp]{\tjudg{\var:\int{\id}}{\var(\out{b}\var)}{(b,\emi)}}
    {\infer[\rax]{\tjudg{\var:\int{\id}}\var{\id}}{}
    & \infer[\rop]{\tjudg{\var:\int{\id}}{\out b\var}{(b,\emi)}}{\infer[\runit]{\tjudg{\var:\int{\id}}{\var}{\unit(\emi)}}{\infer[\rint]{{\tjudg{\var:\int{\id}}\var\emi}}{}}}}}
  & \infer[\rop]{\tjudg{}{\out{a}{\mathsf{II}}}{(a,\int{\id})}}{\infer[\rapp]   
    {\tjudg{}{(\la\vartwo{\vartwo})(\la\varthree{\varthree})}{\unit(\int{\id})}}     
    { \infer[\rabs]
    {\tjudg{}{\la\vartwo{\vartwo}}{\int{\id}\to \unit(\int{\id})}}
    {\infer[\runit]{\tjudg{\vartwo: \int{\id}}{\vartwo}{\unit(\int\id)}}{\infer[\rint]
    {\tjudg{\vartwo: \int{\id}}{\vartwo}{\int{\id}}}
    {\infer[\rax]{\tjudg{\vartwo: \int{\id}}{\vartwo}{\id}}{}}}} & 
    \pi}}
}
\end{array}\]}
    \vspace{-8pt}
    \caption{Type derivation for $\tjudg{}{(\la\var\var(\out{b}{\var}))(\out{a}{\mathsf{II}})}{(ab,\emi)}$. We set $\id\defeq\emi\to\unit(\emi)=\emi\to(\ems,\emi)$.}
    \vspace{-8pt}
    \label{fig:derivation2}
\end{figure*}
\end{example}

\subsubsection*{Relational Reasoning.} 
To prove soundness and completeness of the monadic type system, 
we will need to reason both about expressions and monadic expressions. 
In fact, as long as we are interested in reduction sequences we have to deal both with terms, to start the sequence, and with their (monadic) reducts, to continue 
the computation. Working with monads, we have already seen that we 
can extend the dynamic semantic of $\Lambda^{\mathsf{cbv}}$ to monadic 
computations, for free. Computations, however, come both with a dynamic and 
\emph{static} semantics (\ie types); and if the former semantics (viz. $\tosem$) 
extends to monadic computation (as $\kleisli{\tosem}$), it is not immediately 
clear how to do the same with the latter. In fact, whereas $\tosem$ is a function, 
the static semantics of $\Lambda^{\mathsf{cbv}}$, given by the typing relation 
$\vdash$, is genuinely relational, meaning that we cannot rely on the axioms of a monad 
to extend it.
\subsubsection*{Relational Extensions.} We overcome this problem by relying on the notion of a relational extension of a monad 
\cite{Barr/LMM/1970,DBLP:conf/amast/BackhouseBHMVW91,DBLP:books/algebra/of/programming,carboni-kelly-wood,kawahara1973notes,Hoffman/Cottage-industry/2015}.
Remarkably, relational extensions come with powerful proof techniques, whereby
one extends term-based results (such as subject reduction and expansion, in our case) 
to monadic terms essentially. 
All of that, however, has a price: relational extensions cannot be given for all monads. 
As long as we are interested in `forward properties', such as subject reduction, 
it is enough to give a relaxed notion of relational extensions --- called lax relational extensions --- 
that is available for a large class of monads (such as all the ones seen so far). 
But if we ask for `backward properties', such as subject expansion, then we need the full 
axiomatics of a relational extension. And by a well-known result by Barr~\cite{Barr/LMM/1970}, 
we know that a monad has a relational extension if and only if it is 
\emph{weakly cartesian}~\cite{wc-1,wc-2}. 
From an equational perspective, weakly cartesian monads are defined by
affine theories~\cite{Gautam1957}, meaning that they cannot have equations that duplicate variables. 
This excludes important monads, such as the distribution and powerset monads.  
This way, one has to rely on their `linearization' to obtain well-behaved intersection type systems. In the case of the distribution and powerset monad one does so 
simply by dropping the idempotency equations from their equational theories, thus obtaining the so-called multi-distribution \cite{DBLP:journals/pacmpl/LagoFR21,varacca2003probability,DBLP:journals/mscs/VaraccaW06,DBLP:journals/scp/AvanziniLY20} 
and multiset monads \cite{multisets}. 
It is important to stress that this is not a design issue, but an intrinsic limit of the 
model, as shown by Example~\ref{ex:need-wc} (below in this section).

\subsubsection*{Lifting the Type System.}
The type system in Figure~\ref{fig:mon-system} defines 
a dependent relation
$\trelone \in \prod_{\Gamma} \mathcal{P}(\mathbb{\Lambda}_{\Gamma} \times \mathbb{G})$, 
where $\mathbb{\Lambda} \defeq \mathbb{V} + \mathbb{C}$ 
is the collection of terms of the calculus, and $\mathbb{\Lambda}_{\Gamma}$ 
is the collection of terms with free variables among $\Gamma$.
Notice that $\trelone$ respects syntactic categories, in the sense that 
since $\mathbb{G} = \mathbb{A} + \mathbb{I} + T(\mathbb{I})$, we can see 
$\trelone$ as the sum of three relations 
$\trelone_1 \in \prod_{\Gamma} \mathcal{P}(\mathbb{V}_{\Gamma} \times \mathbb{A})$, 
$\trelone_2 \in \prod_{\Gamma} \mathcal{P}(\mathbb{V}_{\Gamma} \times \mathbb{I})$, 
and $\trelone_3 \in \prod_{\Gamma} \mathcal{P}(\mathbb{C}_{\Gamma} \times T(\mathbb{I}))$. 
In the following, we shall tacitly use this decomposition.  
Moreover, since type soundness and completeness refer to programs, 
we will mostly work with $\trelone$ restricted to closed terms (\ie when $\Gamma$ is empty): 
in that case, $\trelone$ is just an ordinary binary relation. 

When instantiated to monadic types (and closed computations), the relation 
$\trel \subseteq \closed{\mathbb{C}} \times T(\mathbb{I})$ is a 
so-called monadic relation~\cite{DBLP:conf/lics/GavazzoF21}. 
Under suitable conditions on monads, 
monadic relations come with an operation similar to the Kleisli extension 
that allows them to be composed and to regard $\eta$ (seen as a relation) 
as the unit of such an operation.

\begin{definition}[Relational Extension, \hspace{1sp}\cite{Barr/LMM/1970}]
\label{def:relator}
    A \emph{relational extension} 
    of a monad $(\monad, \unit, \multiplication)$ 
    is a family of \emph{monotone} maps 
    $\relator: \rel(A,B) \to \rel(\monad A, \monad B)$ such that:
    \[
    \begin{array}{r@{\hspace{0.1cm}}c@{\hspace{0.1cm}}l@{\hspace{0.8cm}}r@{\hspace{0.1cm}}c@{\hspace{0.1cm}}l}
    \idrel &=& \relator(\idrel)
    &
    \relator(\relone); \relator(\reltwo) 
    &=& \relator(\relone; \reltwo)
    \\
    T(f) &=& \relator(f)
    &
    \dual{\relator(\relone)} &=& \relator(\dual{\relone})
    \\
    \relone; \unit &=& \unit; \relator(\relone)
    &
    \relator(\relator(\relone)); \mu 
    &=& \mu; \relator(\relone)
    \end{array}
    \]
    Replacing $=$ with $\subseteq$, we obtain 
    the notion of a \emph{lax relational extension}.
\end{definition}

Any monad $\monad$ comes with a canonical candidate relational extension: 
its Barr extension $\bext{T}$. Recall that for each relation 
$\relone: A \torel B$, we can regard $\relone$ as a set 
  $\mathcal{G}(\relone) \subseteq A \times B$. In particular, the projections 
  $\pi_1: \mathcal{G}(\relone) \to A$, $\pi_2: \mathcal{G}(\relone) \to B$ give 
  $\relone = \dual{\pi_1}; \pi_2$.

\begin{definition}[Barr Extension]
\label{def:barr-extension}
  The Barr extension $\bext{\monad}$ of $\monad$ 
  is defined as 
  $\bext{\monad}(\relone) = \dual{(\monad(\pi_1))}; \monad(\pi_2).$
  Elementwise, we have
  $
  \phi_1 \mathbin{\bext{\monad}\relone} \phi_2$ 
  iff 
  $$\exists \Phi \in \monad \mathcal{G}(\relone).\ 
  \monad(\pi_1)(\Phi) = \phi_1 \text{ and } \monad(\pi_2)(\Phi) = \phi_2.
  $$
\end{definition}

\begin{example}[Concrete Barr Exts.]
Let $\relone: A \torel B$.
\begin{varenumerate}
  \item For the powerset monad, we have 
    $u \mathbin{\bext{\powerset}{\relone}} v$ iff 
    $\forall x \in u.\exists y \in v.\ x\mathbin{\relone} y$ 
    and $\forall y \in v.\exists x \in u.\ x\mathbin{\relone} y$. 
    A similar definition holds for variations of the powerset monad.
  \item For the distribution monad, we have 
    $\phi_1 \mathbin{\bext{\mathcal{D}}{\relone}} \phi_2$ iff
    there exists $\Phi \in \mathcal{D}(A \times B)$ such that
    $\sum_y \Phi(x,y) = \phi_1(x)$, 
    $\sum_x \Phi(x,y) = \phi_2(y)$, and
    $\Phi(x,y) > 0 \implies x \mathbin{\relone} y 
    $. A similar definition holds for variations of 
    the distribution monad.
  \item The output monad, we have 
    $(a,x) \mathrel{\bext{\writer}(\relone)} (b,y)$ iff 
    $a=b$ and $x \mathrel{R} y$.
\item More generally, if a monad is presented by a theory 
    $(\Sigma, E)$, then we have $t \mathrel{\bext{\tree{\Sigma}}(\relone)} s$ 
    iff $t \eq_E C[x_1, \hdots, x_n]$, $s \eq_E C[y_1, \hdots, y_n]$, and 
    $x_i \mathrel{R} y_i$, for any $i$.
\end{varenumerate}
\end{example}

The Barr extension of a monad is not a relational extension, in general. 
However, the Barr extension of a monad is a relational extension
iff the monad is weakly cartesian~\cite{wc-1,wc-2}.

\begin{theorem}[\cite{Barr/LMM/1970}]\label{thm:barr}
    Recall that a monad $(\monad, \eta, \mu)$ is \emph{weakly cartesian} if
    (i) it preserves weak pullbacks and
    (ii) all naturality squares of $\unit$ and $\mu$ are weak pullbacks. 
    If $\monad$ is weakly cartesian, then its Barr extension is the \emph{unique}
    relational extension of $\monad$. 
   If $\monad$ preserves weak pullbacks, then its Barr extension is a lax relational extension.
\end{theorem}

For brevity, we say that a monad is WC---resp. WP---if it is weakly cartesian---if it preserves weak pullbacks.

\begin{example}
    All the monads seen so far are WP. 
    The output and maybe monad, additionally, are WC, whereas 
    the powerset and distribution monads are not \cite{wc-1}, as 
    naturality squares of their unit are not weak pullbacks. 
    If a monad $T$ is presented by an affine equational theory \cite{Gautam1957} 
    $(\Sigma, E)$, 
    meaning that all equations in $E$ are affine, then 
    it is WC. Consequently, the multiset and multidistribution monad 
    are WC.
\end{example}

Given a monad $\monad$ and 
a monadic relation $\relone \subseteq A \times \monad(B)$, 
we define its Kleisli extension $\kleisli{\relone} \subseteq \monad(A) \times \monad(B)$ 
as $\bext{\monad}(\relone); \mu$. Using Kleisli extension, we define the 
composition of monadic relations $\relone \subseteq A \times \monad(B)$ and 
$\reltwo \subseteq B \times \monad(C)$ as the relation 
$\relone \mcomp \reltwo \subseteq A \times \monad(C)$ defined as 
$\relone; \kleislirel{\reltwo}$. 
If $\monad$ is WC (resp. WP), $\mcomp$ is (lax) associative and has $\eta$ as 
(lax) unit \cite{DBLP:conf/lics/GavazzoF21,hofmann2014monoidal}. 
Using the Kleisli extension we can design
abstract proof techniques ensuring that properties of 
$\trel$ with respect to the
one-step semantics $\tosem$ can be lifted to 
$\kleisli{\trel}$ and $\kleisli{\tosem}$.

\begin{proposition}[\cite{Barr/LMM/1970}]
Let $\relone: A \torel B$ be a relation and $\relator$ be a 
lax relational extension of a monad $T$. Then, 
(i) $\relator(\relone)$ is closed under algebraic operation; 
(ii) $\relator(\relone)$ is closed under monadic binding: 
$R;g  \subseteq f ; \Phi(S)$  implies  $\Phi(R);\kleisli{g} \subseteq \kleisli{f} ; \Phi(S)$. 
\end{proposition}

The next result will be crucial to prove subject expansion.

\begin{proposition}
\label{prop:subject-expansion-lifting}
Let $\monad$ be WC. Then:
    $f; \kleislirel{\relone} \subseteq \reltwo \implies 
    \kleisli{f}; \kleislirel{\relone} \subseteq \kleislirel{\reltwo}$. 
\end{proposition}
In particular, taking both $\relone$ and $\reltwo$ as 
the typability relation $\trel \subseteq \Comps \times \monad(\IntTys)$ and 
as $f$ the one-step semantic function $\tosem: \Comps \to \monad(\Comps)$, we see that 
$\tosem; \kleislirel{\trel} \subseteq \trel$ states that whenever 
we have a term $\tm$ with $\tm\, \tosem\, e$ and a monadic type $M$ with 
$\kleisli{\vdash} e: M$, then
$\vdash \tm: M$.
This is exactly the statement of the subject expansion theorem at the level of term-based 
evaluation that we shall prove in the next section: 
if $\tm \mathrel{\tosem} e$ and $\kleisli{\vdash} e : M$, then
$\vdash \tm: M$.
Prop.~\ref{prop:subject-expansion-lifting} then implies that subject expansion can be 
extended to full monadic reduction $\kleisli{\tosem}$: 
if $e \mathrel{\kleisli{\tosem}} e'$ and $\kleisli{\vdash}e' : M$, then
$\kleisli{\vdash} e : M$.

Obviously, Proposition~\ref{prop:subject-expansion-lifting} still requires us 
to prove $\tosem; \kleislirel{\trel} \subseteq \trel$, and we would like to 
do so syntactically. Although natural, this relational extension is not always possible. The problem lies 
in the fact that if we assign a monadic type $M$ to an element of the form 
$\unit(t)$ via $\kleislirel{\trel}$, there is no guarantee that $t$ itself 
has type $M$. This becomes problematic when dealing with values. 
Since a value $v$ (regarded as a computation) reduces to $\unit(v)$ and 
our monadic type system assigns only types of the form $\unit(I)$ to $v$, 
provided that $\vdash v: I$,
we need to ensure that any type $M$ such that $\kleislirel{\trel} \unit(v) : M$ 
is itself a type of $v$, and hence of the form  $\unit(I)$. This, however, is not always the case.

\begin{example}
\label{ex:need-wc}
    Let us consider the distribution monad $\mathcal{D}$ and 
    recall that its unit maps a point to its Dirac distribution. 
    Let $v$ be a value such that $\vdash v: A$ and $\vdash v: B$, so that 
    $\vdash v: \{A\}$ and $\vdash v: \{B\}$. 
    By the very definition of the monadic type system, the computation induced by the value $v$ 
    can only have monadic types of the form $\unit(I)$ (\ie Dirac distributions 
    $1 \cdot I$). 
    Yet, the lifted relation $\kleislirel{\vdash}$ gives 
    $\kleislirel{\trel} \unit(v) : \frac{1}{2} \cdot \{A\} + \frac{1}{2} \cdot \{B\}$, since 
    $\unit(v) = 1 \cdot v = \frac{1}{2} \cdot v + \frac{1}{2} \cdot v$ 
    and $\vdash v: \{A\}$ and $\vdash v: \{B\}$ entail 
    $\kleislirel{\vdash} 1\cdot v: 1 \cdot \{A\}$ and $\kleislirel{\vdash} 1\cdot v: 1 \cdot \{B\}$. 
    Consequently, we have $\kleislirel{\vdash} \unit(v) : \frac{1}{2} \cdot \{A\} + \frac{1}{2} \cdot \{B\}$ 
    but we cannot have $\vdash v : \frac{1}{2} \cdot \{A\} + \frac{1}{2} \cdot \{B\}$.
\end{example}

The ultimate source of the problem outlined in the above example 
is that the unit of $\mathcal{D}$ is not weakly cartesian.

\begin{proposition}
\label{prop:subject-expansion-eta} 
  Let $\monad$ be WC. For any monadic relation  
    $\relone \subseteq A \times \monad(B)$, we have 
    $\eta; \kleislirel{\relone} = \relone$. 
    In particular, $\eta; \kleislirel{\trel} = {\trel}$.
\end{proposition}


The techniques seen so far have been designed to prove subject expansion 
of the monadic type system. As expected, we are also interested in proving 
subject reduction and thus it is natural to design similar 
proof techniques in that setting. This can be easily done following the 
same path as for subject expansion, but with a main difference: subject reduction 
does not require the unit of the monad to be WC, and hence 
subject reduction results can be proved for a much larger class of monads.\footnote{
Notice that even if $\unit$ is not WC, we still have 
$\relone; \eta \subseteq \eta; \relator\relone$ (but not the other inclusion, which is 
crucial in Proposition~\ref{prop:subject-expansion-eta}).}
 \begin{proposition}
\label{prop:subject-reduction-lifting}
Let $\monad$ be WP. Then:
    $\relone^{\circ}; f \subseteq {\kleislirel{\relone}}^{\circ} \implies 
    {\kleislirel{\relone}}^{\circ}; \kleisli{f} \subseteq {\kleislirel{\relone}}^{\circ}$. 
\end{proposition}
In particular, we can instantiate Proposition~\ref{prop:subject-reduction-lifting} 
with $\relone$ and $f$ as $\trel$ and $\tosem$, hence obtaining 
$\trel^{\circ}; \tosem \subseteq {\kleislirel{\trel}}^{\circ} \implies 
    {\kleislirel{\trel}}^{\circ}; \kleisli{\tosem} \subseteq {\kleislirel{\trel}}^{\circ}$,
meaning that whenever subject reduction holds at the level of terms
(\ie $\vdash \tm: M \mathrel{\&} \tm \mathrel{\tosem} e \implies \kleisli{\vdash} e : M$), 
then it holds at the level of their (monadic) evaluation 
(\ie $\kleislirel{\vdash} e: M$ and $e \mathrel{\kleisli{\tosem}} e'$ 
implies $\kleisli{\vdash} e' : M$).

\subsubsection*{Soundness.}
The proof of soundness of the type system consists in showing:
\begin{varenumerate}
    \item \emph{Subject reduction:} types are preserved by reduction.
    \item \emph{Termination:} all typable terms terminate.
\end{varenumerate}
As we have seen, by Proposition~\ref{prop:subject-reduction-lifting}, it is sufficient to prove subject reduction 
with respect to the single-step, term-based reduction $\tosem$. This is
done by induction on the structure of evaluation contexts, with the help of a substitution lemma, proved by induction on the structure of terms.
\begin{proposition}[Subject Reduction]\label{prop:sub-red}
\mbox{}
Let $\monad$ be WP. Then:
\begin{varenumerate}
    \item  Let $\tm$ be a closed $\lambda$-term. If $\tjudg{}{\tm}{\monty}$ and $\tm\,\tosem\,\etm$, then 
	    $\ltjudg{}{\etm}{\monty}$.
	\item Let $e$ be a monadic closed $\lambda$-term. If $\kleislirel{\vdash} e: M$ 
	    and $e \mathrel{\kleisli{\tosem}} e'$, then $\kleisli{\vdash} e' : M$.
	\end{varenumerate}
\end{proposition}
Proving termination instead needs more work.
\subsubsection*{Effectful Observations.} Knowing that typing is preserved by reduction, it remains to show that 
whenever a computation $t$ has type $M$, its observable operational 
behaviour is fully captured by $M$. 
In the pure case, such a behaviour is just termination, so that 
one usually shows that typable terms terminate. In the effectful setting, 
termination can be given in many forms. First, if effects capture some forms of 
nondeterminism, meaning that elements in $\monad \mathbb{\Lambda}$ may have more 
than one element in their support, then termination can be divided into 
\emph{may} or \emph{must} termination (\ie whether term reaches monadic expressions 
with one, at least, or all values in their support). In both of these cases, termination 
remains a boolean notion (viz. a predicate). To account for effects it is natural to 
ask not only whether a computation terminates, but also which effects are produced during 
evaluation (\eg what is stored in memory locations, which are the printed outputs, the cost of the 
computation, etc). A further option is to make termination effectful itself, a well-known 
example of effectful termination being almost-sure termination (\ie probability of convergence). 
Such notions are usually infinitary and require non-boolean reasoning. 

Since here we deal with the finitary case, we agree to observe must termination of computations as well as
the effects produced during their evaluation. In the next section, we shall deal 
with infinitary evaluation and, consequently, with effectful termination. 
Let us begin by formalising how to observe effects. In a monadic setting, 
it is customary \cite{gavazzo/LICS/2017,lago-gavazzo-ictcs,lago-gavazzo-tcs-20,DBLP:journals/toplas/SimpsonV20}, to model (effectful) observables as elements of 
$\monad(X)$, where $X$ is what is observable of 
expressions. As we are interested in must termination, only values 
are observable, and, moreover, they cannot be scrutinized further 
 (\ie we observe that a computation 
gives a result  (a value), but we cannot inspect such a result\footnote{This is standard in weak, untyped $\l$-calculi. One could add constants, such as booleans, or numerals, and then observe their shape, in a straightforward way.}). 

\begin{definition}
\label{def:observation}
    We define the observation function for monadic objects in $T(X)$ as 
    $\obsfun_{X} \defeq T(!_{X}): T(X) \to T(1)$, where 
    $1 \defeq \{\star\}$ and $!_{X}: X \to 1$ is the unique arrow collapsing all 
    the elements of $X$ to $\star$. 
    We extend $\obsfun_{\mathbb{\Lambda}}$ to a partial function 
    on terms by stipulating $\obsfun_{\mathbb{\Lambda}}(t) \defeq \obsfun_{\mathbb{\Lambda}}(e)$, 
    provided that $ \sem t = e$.

\end{definition}
 As usual, we omit subscripts whenever possible, writing 
 $\obs{e}$, $\obs{M}$, etc.

\begin{example}[Concrete Observations]
\begin{varenumerate}
\item[]
\item The output, or writer, monad $\mathcal{W}$ has a notion of observation $\obsfun:\mathcal{W}(X)\to W$, if $W$ is the underlying monoid of words. This is immediate to see because $\mathcal{W}(1)\defeq W\times \{\star\}\cong W$. Then, we have that $\obsfun((w,x))\defeq w$. This means that what we can observe is the string that has been printed on the output buffer during the computation.
\item The partiality monad $\mathcal{E}$ provides a binary notion of observation, indeed $\obsfun:\mathcal{E}(X)\to \{\star,\bot\}$. This is actually the way in which one could observe divergence.
\item The powerset monad $\mathcal{P}$ comes with the natural notion of must termination, since $\obsfun:\mathcal{P}(X)\to \{\emptyset,1\}$.
\end{varenumerate}
\end{example}

\begin{remark}
    According to Definition~\ref{def:observation}, the observable effects produced by a computation are elements of $T(1)$. This certainly works well for some effects and monads, 
    such as output and cost, but it may be unusual for others. For instance, probabilistic 
    nondeterminism is usually modelled using (variations of) the sub-distribution monad 
    $\mathcal{D}$ and, since $\mathcal{D}(1) \cong [0,1]$, it is natural to 
    interpret elements in the latter set as actual probabilities of events (such as the probability 
    of termination, in our case). However, we have already seen that it is 
    simply not possible to have well-behaved forms of intersection types working 
    with $\mathcal{D}$, and that we can overcome that issue by working with 
    the multi-sub-distribution monad $\mathcal{M}$. Unfortunately, 
    $\mathcal{M}(1) \not \cong [0,1]$, although we would still like 
    to think about the observable behaviour of a program as its probability of 
    convergence. 
    This is not a big issue since it does not take much to realise that 
    our analysis of observations works \emph{mutatis mutandis} if we replace 
    $T(1)$ with $S(1)$, where $S$ is another monad such that there is a monad 
    morphism $\nu: T \Rightarrow S$. This way, for instance, even if modelling 
    static and dynamic semantics in terms of $\mathcal{M}$, 
    we can regard 
    $\obs{t}$ as the probability convergence of convergence of $t$, due 
    to the monad morphism $\nu: \mathcal{M} \Rightarrow \mathcal{D}$ 
    collapsing multi-sub-distributions into ordinary sub-distributions.
\end{remark}

\subsubsection*{Termination by Logical Relations.} Our goal is now to prove that whenever a term $t$ has type $M$ then: (i) $t$ must 
terminate, and (ii) the observable behaviour of $t$, \ie, $\obs{t}$, is 
fully described by 
$M$. That is, $\obs{t} = \obs{M}$. 
To achieve such a goal, we define a logical relation $\models$ between (closed) terms and types 
acting as the semantic interpretation of $\vdash$ in such a way that ${\vdash} = {\models}$ 
and $\models t: M$ implies $t \Downarrow e$ (must termination) and 
$\obs{e} = \obs{M}$. Remarkably, such a logical relation makes crucial use of 
the Barr extension of $\monad$.

\begin{definition}
We define the logical semantics of $\trel$ (restricted to closed expressions) 
as the relation ${\trelsem} \defeq {\trelsem_{\mathbb{A}}} + {\trelsem_{\mathbb{I}}} + 
{\trelsem_{\mathbb{M}}}$
that inductively refines $\trel$ (\ie $\trelsem\,\subseteq\trel$) as follows, where 
we use the notation $\widehat{\trelsem} e : M$ in place of 
$e  \mathrel{\bextt{\trelsem_{\mathbb{I}}}} M$. 
\[
\begin{array}{l@{\hspace{0.1cm}}c@{\hspace{0.1cm}}l}
    \trelsem_{\mathbb{A}} v: I \to M 
    &\mbox{ iff }& \forall w.\, \trelsem_{\mathbb{I}} w: I \mbox{\,implies\,} \trelsem_{\mathbb{M}} vw: M
    \\
    \trelsem_{\mathbb{I}} v: \{A_1,..,A_n\} 
    &\mbox{ iff }& \forall i.\, \trelsem_{\mathbb{A}} v: A_i
    \\
    \trelsem_{\mathbb{M}} t: M 
    &\mbox{ iff }& \mathrel{\widehat{\trelsem_{\mathbb{I}}}} \sem t : M
\end{array}\]
\end{definition}
As usual, we omit subscripts whenever unambiguous. 
We first show that, indeed, $\trelsem$ ensures the desired property.

\begin{lemma}\label{l:semlr}
$\trelsem t: M$ implies $\exists e\text{ such that }\sem t = e$ and $\obsfun(e) = \obsfun(M)$.
\end{lemma}
Then, we prove the soundness of our type system showing that $\trelsem$ and $\vdash$ coincide.

\begin{proposition}[Soundness]\label{prop:soundness-lr}
    ${\trel} = {\trelsem}$.
\end{proposition}

\subsubsection*{Completeness.} Having proved soundness of our type system, 
we now move on to completeness, meaning that normalising terms are typable. 
The proof of completeness follows the usual pattern for intersection types, and makes crucial use of subject expansion. The proof of subject expansion 
is divided into two parts: first, we prove subject expansion with respect to the single-step reduction 
on terms and then extend such a result to monadic terms and monadic reduction 
relying on Proposition~\ref{prop:subject-expansion-lifting}.
Concerning the first part, we would like to prove subject expansion by induction on the structure of evaluation contexts (after having proved a straightforward anti-substitution lemma). However, the statement is not true in general.
\begin{example}
    Let us consider the multidistribution monad and the binary operation $\oplus_1$, \ie the first projection. Let us consider the reduction $\la\var\var\oplus_1\Omega\mapsto 1\cdot\la\var\var$. Even if $\kleisli\vdash 1\cdot\la\var\var:1\cdot\monty $, it is not possible to type $\la\var\var\oplus_1\Omega$ with rule $\rop$, because of course there is no way of typing $\Omega$.
\end{example}
Then, we need a restriction on our calculus, this time about operations. We allow only operations $\op{\tm_1,\ldots,\tm_n}$ that do \emph{not} erase their arguments, \ie for which $\supp{\gefop{\unit(\tm_1),\ldots,\unit(\tm_n)}}=\{\tm_1,\ldots,\tm_n\}$. In terms of equational theories, this is guaranteed by considering \emph{linear} theories. We already anticipate that we will be able to remove this restriction in the next section, by the use of infinitary means. 
\begin{proposition}[Subject Expansion]
\label{prop:subject-expansion}\mbox{} Let $\monad$ be WC. Then:
\begin{enumerate}
    \item 	If $\tm \mathrel{\tosem} e$ and $\kleisli{\vdash} e : M$, 
        then $\vdash \tm: M$.
	\item If $e \mathrel{\kleisli{\tosem}} e'$ and $\kleisli{\vdash}e' : M$, 
	    then $\kleisli{\vdash} e : M$.
\end{enumerate}
\end{proposition}

Proposition~\ref{prop:subject-expansion}, together with the fact that monadic values can always be typed, 
gives the completeness of the type system.

\begin{theorem}[Completeness]
\label{thm:completeness}
	If $\sem\tm=\etm$, then there exists a monadic type  $\monty$ such that $\tjudg{}{\tm}{\monty}$.
\end{theorem}
Soundness and completeness together provide a characterization of finitary effectful termination via typability with intersection types.
\begin{corollary}[Characterization]
    The following clauses are equivalent:
    \begin{enumerate}
        \item Effectful termination: $\obs{\sem t}=o$.
        \item Typability: there exists $\monty$, such that $\tjudg{}{\tm}{\monty}$ and $\obs\monty=o$.
    \end{enumerate}
\end{corollary}

\section{Infinitary Effectful Semantics}\label{sec:infinitary}
In this section, we extend the type system of Section~\ref{sec:mon-types} 
to account for infinitary behaviours. To do so, we 
require monads to have enough structure to support 
such behaviours. A standard approach to do that is 
by requiring suitable order-theoretic enrichments. 
Here, we consider monads whose Kleisli category is enriched
in the category of directed complete pointed 
partial order (\dcppo{s}) \cite{Kelly/EnrichedCats,AbramskyJung/DomainTheory/1994}, 
but in order to maintain the paper as self-contained as possible, 
we use the following more concrete (and restricted) definition.

\begin{definition}
A monad $(T, \eta, \bindsymbol)$ is 
    \dcppo{-}ordered if, for any set $A$, we have a 
    \dcppo{} $(T(A), \cpoleq_A, \bot_A)$ such that 
    the bind operator is strict and continuous in both 
    arguments.
\end{definition}

As usual, we omit subscripts whenever unambiguous. Notice that 
if $T$ is \dcppo{-}ordered, then all its algebraic 
operations are strict and continuous. 

\begin{example}
Both the multiset and multidistribution monad can be 
turned into \dcppo{-}ordered monads by simply adding 
a zero-ary operation symbol $\bot$ to their equational theories. 
Semantically, $\bot$ corresponds to the empty multiset and 
multidistribution, respectively. 
\end{example}

From now on, we tacitly work with an arbitrary but fixed 
\dcppo{-}ordered monad $(T, \eta, \bindsymbol)$.

\subsubsection*{Infinitary Typing.}
Extending the monadic type system to the infinitary case 
is straightforward. We simply add the typing rule 
\[\infer[\rbot]{\Gamma \vdash t: \bot}{}\] allowing to type any computation 
with the total uninformative type $\bot$. Consequently, 
we can assign several types to each term.

\begin{example} We provide in Fig.~\ref{fig:derivationthree} the type derivation for the term $\tjudg{}{(\la\var\var\var)(\mathsf{II}\oplus \Omega)}{\frac{1}{2}\emi}$, again a simple variation on the theme of the previous examples. One can notice that we are able to type it, even if clearly the term does \emph{not} converge. Its type $\frac{1}{2}\emi$ says exactly that: the probability of convergence is $\obs{\frac{1}{2}\emi}=\frac{1}{2}$.
\begin{figure*}[t]
{\scriptsize\[
\begin{array}{c}
\Phi:\infer[\rabs]
    {\tjudg{}{\la\var{\var\var}}{\int{\id}\to\unit(\emi)}}
    {\infer[\rapp]{\tjudg{\var:\int{\id}}{\var\var}{\unit(\emi)}}
    {\infer[\rax]{\tjudg{\var:\int{\id}}\var{\id}}{}
    &\infer[\runit]{\tjudg{\var:\int{\id}}\var{\unit(\emi)}}{\infer[\rint]{{\tjudg{\var:\int{\id}}\var\emi}}{}}}}
 \\\hhline{-}\\[2mm]
 \Psi:\infer[\rop]{\tjudg{}{\mathsf{II}\oplus\Omega}{\frac{1}{2}\int{\id}}}{
  \infer[\rapp]   
    {\tjudg{}{(\la\vartwo{\vartwo})(\la\varthree{\varthree})}{1\int{\id}}}
    { \infer[\rabs]
    {\tjudg{}{\la\vartwo{\vartwo}}{\int{\id}\to 1\int{\id}}}
    {\infer[\runit]{\tjudg{\vartwo:\int\id}{\vartwo}{1\int\id}}{\infer[\rint]
    {\tjudg{\vartwo: \int{\id}}{\vartwo}{\int{\id}}}
    {\infer[\rax]{\tjudg{\vartwo: \int{\id}}{\vartwo}{\id}}{}}}} & 
    \infer[\runit]{\tjudg{}{\la\varthree\varthree}{1\int\id}}{\infer[\rint]{\tjudg{}{\la\varthree{\varthree}}{\int{\id}}}
    {    \infer[\rabs]{\tjudg{}{\la\varthree{\varthree}}{\id}}{ \infer[\runit]{\tjudg{\varthree:\emi}\varthree{\unit(\emi)}}{  
    \infer[\rint]{\tjudg{\varthree:\emi}\varthree\emi}{}}}}
    }} & \infer[\rbot]{\tjudg{}{\Omega}{\bot}}{}
  }
\\\hhline{-}\\[2mm]
  \infer[\rapp]
  {\tjudg{}{(\la\var\var\var)(\mathsf{II}\oplus\Omega)}{\frac{1}{2}\emi}}
  {\Phi \quad &\quad \Psi}
\end{array}
\]}
    \vspace{-8pt}
    \caption{Type derivation for $\tjudg{}{(\la\var\var\var)(\mathsf{II}\oplus \Omega)}{\frac{1}{2}\emi}$. $\id\defeq\emi\to\unit(\emi)=\emi\to 1{\cdot}\emi$.}
    \vspace{-8pt}
    \label{fig:derivationthree}
\end{figure*}
\end{example}

Nonetheless, we can think about type derivations (with occurrences of the $\rbot$ rule) as approximations of the semantic
content of a computation, the latter being reached only at the limit. Moreover, the set of the observations $\mathcal{O}(\tm)$ of a term $t$, defined as the collection of all the observations $\obs M$ for $\vdash t : M$ is directed. 
Consequently, we can associate to each $t$ a more informative observation obtained through types
given as $O(\tm) \defeq \join \mathcal{O}(\tm)$. Notice that even if $O(\tm)$ 
is a valid observation, there may be no (necessarily finite) derivation  $\pi \mathrel{\pof} {\vdash t: \monty}$ such that $\obs\monty=O(\tm)$.

\subsubsection*{Infinitary Operational Semantics.}
A standard approach to deal with infinitary effectful semantics 
consists in defining a monadic evaluation function mapping computations 
to monadic values. To capture forms of convergence in the limit, 
such a function is defined as a suitable least upper bound of maps 
evaluating computations for a fixed number of steps. We implement this 
strategy building upon the definition of $\tosem$. 

\begin{definition}[Approximate Operational Semantics]
    Let $\phi: \Comps \to T(\Vals)$ mapping values $v$ (as elements 
    in $\Comps$) to $\eta(v)$ and all other terms $t$ to $\bot$. 
    Then, we define the $\mathbb{N}$-indexed family of maps 
    $\sem{-}^n: \closed{\Comps} \to \monad(\closed{\Vals})$ by 
    $\sem{t}^0 \defeq \bot$, and
    $\sem{t}^n \defeq e \bind \phi$, if $n > 0$ and $t \mathrel{\tosem^n} e$.
\end{definition}


\begin{lemma}[\cite{gavazzo/LICS/2017}]
    For any closed computation $t$, the sequence $\{\sem{t}^n\}_{n \geq 0}$ 
    forms a directed set (an $\omega$-chain, actually).
\end{lemma}

Consequently, we define (overriding the previous finitary definition) $\sem{t} = \lub_n \sem{t}^n$. Notice that this also gives 
a straightforward way to extend the observation function $\obsfun$ (on terms) 
to the infinitary setting. We simply define 
$\obs{t} \defeq T(!)(\sem{t})$, with $!: \Vals \to 1$ be as before. Moreover, 
since $T$ is \dcppo{-}enriched, $\obsfun$ is continuous (and thus monotone). 
In particular, we can define a bounded observation function 
as $\obsn{t}{n} \defeq T(!)(\sem{t}^n)$ and see that 
$\obs{t} = \lub_n \obsn{t}{n}$.
We now have all the ingredients needed to extend our characterization to the 
infinitary setting.
\section{Characterizing Infinitary Behaviors}\label{sec:infchar}

In this section, we extend the soundness and completeness results 
previously seen to the infinitary setting. 
Remarkably, most of the proofs given in the finitary case, such as 
those of subject reduction and expansion, scale to the infinitary case. This is no coincidence but a main strength 
of the abstract relational approach that we have developed in the previous 
part of this work.

\subsubsection*{Soundness.}
As in the finitary case, we have subject reduction for WP monads.

\begin{proposition}[Subject Reduction, Infinitary]
\label{prop:subject-reduction-infinitary}\mbox{} Let $\monad$ be WP. Then:
\begin{enumerate}
    \item  Let $\tm$ be a closed $\lambda$-term. If $\tjudg{}{\tm}{\monty}$ and $\tm\,\tosem\etm$, then 
	    $\kleisli{\vdash}{\etm}:{\monty}$.
	\item Let $e$ be a monadic closed $\lambda$-term. If $\kleislirel{\vdash} e: M$ 
	    and $e \mathrel{\kleisli{\tosem}} e'$, then $\kleisli{\vdash} e' : M$.
	\end{enumerate}
\end{proposition}

Notice that Proposition~\ref{prop:subject-reduction-infinitary} is given relying on the 
Barr extension of $T$ which, by its very definition, does not take into account 
the order $\cpoleq$ induced by $T$. In particular, whenever 
we have $\kleisli{\vdash} e : M$, then $e$ and $M$ must have the same 
effectful behaviour. This means that as long as we stick with $\bextt{\vdash}$, 
it is simply not possible to extend subject 
reduction (and thus soundness) to the full evaluation $\sem{-}$, the latter being 
infinitary. 
Consequently, contrary to the finitary case, there is no hope to prove that whenever 
$\vdash t: M$, then $M$ encodes the whole observable behaviour of $t$, i.e. 
$\obs{\sem{t}}$. What we can show, however, is that $M$ provides an approximation 
of such a behaviour, and that the limit of such approximations is 
precisely the operational behaviour of $t$. The right tool to achieve such a goal,
is an \emph{ordered} version of the Barr extension \cite{Huge-Jacobs/Simulations-in-coalgebra/2004}.

\begin{definition}[Right Barr Extension]
    Given $\relone \subseteq A \times B$, we define its \emph{right Barr extension} 
    $\widehat{T}_{\cpogeq}(\relone) \subseteq T(A) \times T(B)$ as 
    $\widehat{T}_{\cpogeq}(\relone) \defeq \bextt{\relone}; {\cpogeq}$.
\end{definition}

\begin{proposition}[\cite{Huge-Jacobs/Simulations-in-coalgebra/2004}]
\label{prop:right-barr-extension-is-a-lax-extension}
If $\monad$ is WP, then
    $\widehat{T}_{\cpogeq}$ is a lax relational extension.
\end{proposition}

Using the right Barr extension, we define the logical relation 
interpreting $\vdash$.

\begin{definition}[Infinitary Logical Relation]
We define the logical semantics of $\trel$ (restricted to closed expressions) 
as the relation ${\trelsem} \defeq {\trelsem_{\mathbb{A}}} + {\trelsem_{\mathbb{I}}} + 
{\trelsem_{\mathbb{M}}}$
that inductively refines $\trel$ (\ie $\trelsem\,\subseteq\trel$) as follows, where 
we use the notation $\widetilde{\trelsem}\; e : M$ in place of 
$e  \mathrel{\widehat{T}_{\cpogeq}(\trelsem)} M$.
\[
\begin{array}{l@{\hspace{0.05cm}}c@{\hspace{0.05cm}}l}
    \trelsem_{\mathbb{A}} v: I \to M
    &\mbox{ iff }& \forall w.\, \trelsem_{\mathbb{I}} w: I \mbox{\,implies\,} \trelsem_{\mathbb{M}} vw: M
    \\
    \trelsem_{\mathbb{I}} v: \{A_1, ..., A_n\} 
    &\mbox{ iff }& \forall i.\, \trelsem_{\mathbb{A}} v: A_i
    \\
    \trelsem_{\mathbb{M}} t: M 
    &\mbox{ iff }& \widetilde{\trelsem}\; \sem{t} : M 
\end{array}
\]
\end{definition}
As usual, we omit subscripts whenever unambiguous. 

\begin{lemma}\label{l:semlrinf}
$\trelsem t: M$ implies $\obs{\sem{t}} \cpogeq \obsfun(M)$.
\end{lemma}
As in the finitary case, we prove the soundness of our type system showing that $\trelsem$ and $\vdash$ coincide.


\begin{proposition}[Soundness]\label{prop:soundnessI}
If $\monad$ is WP, then
    ${\trel} = {\trelsem}$.
\end{proposition}


\subsubsection*{Completeness.}
As in the finitary case, completeness is proved via subject expansion. This latter result, in turn, is obtained exactly as in the finitary case. The only difference is that we are able to drop the constraint about non-erasing operations. Indeed, this time we can type erased (and thus possibly diverging) arguments with the rule $\rbot$.
\begin{proposition}[Subject Expansion, Infinitary]
\label{prop:subject-expansion-inf}\mbox{} Let $\monad$ be WC. Then:
\begin{enumerate}
    \item 	If $\tm \mathrel{\tosem} e$ and $\kleisli{\vdash} e : M$, 
        then $\vdash \tm: M$.
	\item If $e \mathrel{\kleisli{\tosem}} e'$ and $\kleisli{\vdash}e' : M$, 
	    then $\kleisli{\vdash} e : M$.
\end{enumerate}
\end{proposition}
Then, we are able to prove approximate completeness, by finitary means. 
\begin{theorem}[Approximate Completeness]\label{thm:completenessI}
Let $\tm$ be a closed $\l$-term. Then, for each $k\geq 0$, there exist $\pi_k\pof\tjudg{}\tm{\monty_k}$ such that $\obs{\monty_k}=\obsn{\tm}{k}$.
\end{theorem}
Finally, we can claim the full characterization of the infinitary effectful behavior of any program by the way of our intersection type system. 
\begin{corollary}[Characterization]\label{c:char-inf}
Let $\tm$ be a closed $\l$-term. Then $O(\tm)=\obs{\tm}$.
\end{corollary}

\begin{figure*}[t]
{\scriptsize\[
\begin{array}{c}
\Phi:\infer[\rabs]
    {\tjudg{}{\la\var\checkmark(\var\var)}{\int{\id}\to1{\cdot}(1,\emi)}}
    {\infer[\rop]{\tjudg{\var:\int{\id}}{\checkmark(\var\var)}{1{\cdot}(1,\emi)}}{\infer[\rapp]{\tjudg{\var:\int{\id}}{\var\var}{\unit(\emi)}}
    {\infer[\rax]{\tjudg{\var:\int{\id}}\var{\id}}{}
    &\infer[\runit]{\tjudg{\var:\int{\id}}\var{\unit(\emi)}}{\infer[\rint]{{\tjudg{\var:\int{\id}}\var\emi}}{}}}}}
	\\\hhline{-}\\[2mm]
\Psi:\infer[\rop]{\tjudg{}{\checkmark(\mathsf{II})\oplus \mathsf{I}}{\frac{1}{2}{\cdot}(1,\int\id),\frac{1}{2}{\cdot}(0,\int\id)}}{
  \infer[\rop]{\tjudg{}{\checkmark(\mathsf{II})}{1{\cdot}(1,\int\id)}}{\infer[\rapp]   
    {\tjudg{}{(\la\vartwo{\vartwo})(\la\varthree{\varthree})}{1{\cdot}(0,\int\id)}}
    { \infer[\rabs]
    {\tjudg{}{\la\vartwo{\vartwo}}{\int{\id}\to 1{\cdot}(0,\int\id)}}
    {\infer[\runit]{\tjudg{\vartwo:\int\id}{\vartwo}{1{\cdot}(0,\int\id)}}{\infer[\rint]
    {\tjudg{\vartwo: \int{\id}}{\vartwo}{\int{\id}}}
    {\infer[\rax]{\tjudg{\vartwo: \int{\id}}{\vartwo}{\id}}{}}}} & 
    \infer[\runit]{\tjudg{}{\la\varthree\varthree}{1{\cdot}(0,\int\id)}}{\infer[\rint]{\tjudg{}{\la\varthree{\varthree}}{\int{\id}}}
    {    \infer[\rabs]{\tjudg{}{\la\varthree{\varthree}}{\id}}{ \infer[\runit]{\tjudg{\varthree:\emi}\varthree{\unit(\emi)}}{  
    \infer[\rint]{\tjudg{\varthree:\emi}\varthree\emi}{}}}}
    }}} & \qquad \hspace{-55pt}\infer*{\tjudg{}{\mathsf{I}}{1{\cdot}(0,\int\id)}}{}}\\\hhline{-}\\[2mm]
    \infer[\rop]{\tjudg{}{\checkmark((\la\var\checkmark(\var\var))(\checkmark(\mathsf{II})\oplus \mathsf{I}))}{\frac{1}{2}{\cdot}(3,\emi),\frac{1}{2}{\cdot}(2,\emi)}}{\infer[\rapp]
  {\tjudg{}{(\la\var\checkmark(\var\var))(\checkmark(\mathsf{II})\oplus \mathsf{I})}{\frac{1}{2}{\cdot}(2,\emi),\frac{1}{2}{\cdot}(1,\emi)}}
  { \Phi \quad &\quad \Psi}}
\end{array}
\]}
    \vspace{-8pt}
    \caption{Type derivation for $\tjudg{}{\checkmark((\la\var\checkmark(\var\var))(\checkmark(\mathsf{II})\oplus 
    \mathsf{I}))}{\frac{1}{2}{\cdot}(3,\emi),\frac{1}{2}{\cdot}(2,\emi)}$. We set $\id\defeq\emi\to\unit(\emi)=\emi\to 1{\cdot}(0,\emi)$.}
    \vspace{-8pt}
    \label{fig:derivationfour}
\end{figure*}

\begin{example}
As a concluding example, in Fig.~\ref{fig:derivationfour} we show the type derivation for the term $\tjudg{}{\checkmark((\la\var\checkmark(\var\var))(\checkmark(\mathsf{II})\oplus \mathsf{I}))}{\frac{1}{2}{\cdot}(3,\emi),\frac{1}{2}{\cdot}(2,\emi)}$. Please notice that this example is built on the composition of two different monads: cost and multidistribution. This way, we show how we are able to handle computational effects in a modular way. It is easy to verify that the same would have been possible, \eg, for cost and multipowerset. $\obs{\frac{1}{2}{\cdot}(3,\emi),\frac{1}{2}{\cdot}(2,\emi)}=\frac{1}{2}{\cdot}3,\frac{1}{2}{\cdot}2$, which is a \emph{distribution} on \emph{costs}. Taking its expected value, one can indeed obtain the average cost of the computation. One can build an example that actually uses infinitely many types (and derivations) on the same line of the one presented in~\cite{DBLP:journals/pacmpl/LagoFR21}.
\end{example}

\section{Conclusion}\label{sec:conclusion}
In this paper, we have proposed the first intersection type system able to characterize the effectful behavior for terms of the $\l$-calculus enriched with algebraic operations. In particular, we are able to do that parametrically with respect to the underlying monad. Moreover, having presented effects as algebraic theories, it is possible to compose effects relying both on the sum and tensor of algebraic theories. Since effectful behaviors are often observed at the limit, we had to deal with infinitary constructions. Technically speaking, relational reasoning was the main tool exploited to obtain our result in an abstract and modular way.
\subsubsection*{Perspectives.} This work opens several research directions:
\begin{list}{$\bullet$}{\leftmargin=1em \itemindent=0em}
    \item \emph{Quantitative Cost Analyses:} Idempotent intersection types are qualitative in nature because they are not able to track the use of resources, such as time or space, during the evaluation. Turning intersections (\ie sets) into multisets is enough to measure the precise cost of the evaluation of typed terms, while maintaining the correctness of the type system~\cite{deCarvalho18,DBLP:journals/jfp/AccattoliGK20,DBLP:journals/pacmpl/AccattoliLV22}. Extending this machinery to the effectful setting would be very interesting, although not trivial. While there are standard notions of monadic costs (\eg the \emph{average} cost in the probabilistic setting, or the \emph{maximum} cost in must nondeterminism), it is not clear how to devise the type system to capture them. In the probabilistic case, for example, some additional information had to be stored inside types to correctly compute the average number of steps~\cite{DBLP:journals/pacmpl/LagoFR21}. Very recently some investigations on the state and the exception monad (featuring also handling) have appeared~\cite{DBLP:conf/wollic/AlvesKR23,excTLLA}, but the design of the type systems seems ad-hoc and not easy to generalize.
    \item \emph{Higher-Order Model Checking:} Model checking of higher-order recursion schemes has been proven decidable by Ong in 2006~\cite{DBLP:conf/lics/Ong06}. Since then, several papers dissected the original result and gave other proof methods and model checking algorithms. Among them, Kobayashi and coauthors developed type-theoretic techniques based on intersection types~\cite{DBLP:conf/lics/KobayashiO09,DBLP:conf/popl/Kobayashi09}. While the literature contains results about model checking higher-order programs enriched with specific effects, such as probability~\cite{DBLP:journals/lmcs/KobayashiLG20} or nondeterminism~\cite{DBLP:conf/fossacs/Tsukada014}, no general method covering families of computational effects is known. Indeed, we would like to investigate if our type system could guide the synthesis of model checking algorithms in the style of~\cite{DBLP:conf/fossacs/Tsukada014}. Since the problem has been proved in general undecidable in the effectful setting, \eg in the case of the sub-distribution monad~\cite{DBLP:journals/lmcs/KobayashiLG20}, one would need of course to restrict the class of monads in order to recover decidability.
    \item \emph{Adding Coinduction:} Our type system is not able to deal with coinductive properties, such as productivity, or with coinductive effects, like the output of streams. We would like to enhance the type system with coinductive types/rules in order to capture these kind of properties. Since the type system is somehow modelled on top of operational semantics, this would require to change it as well. We mention that very recently some works covering coinduction have appeared, but limited to the pure $\l$-calculus, and carried on in the non-idempotent setting~\cite{DBLP:conf/lics/Vial17,DBLP:conf/lics/Vial18}. 
\end{list}

\subsubsection*{Acknowledgments.} To Ugo Dal Lago, who set the basis for this work.

\bibliographystyle{splncs04}
\bibliography{main}

\appendix

\section{Omitted Proofs of Section~\ref{sec:mon-types}}

\subsubsection*{Proposition~\ref{prop:subject-expansion-lifting}.} Let $\monad$ be WC. Then:
    $f; \kleislirel{\relone} \subseteq \reltwo \implies 
    \kleisli{f}; \kleislirel{\relone} \subseteq \kleislirel{\reltwo}$. 
    \begin{proof}
\begin{align*}
     \kleisli{f}; \kleislirel{\relone} &= \monad f; \mu; \relator \relone; \mu 
     = \monad f; \relator\relator \relone; \mu;\mu
     = \monad f; \relator\relator \relone; \monad \mu;\mu
     = \relator (f; \relator \relone; \mu);\mu
     \\&=  \relator (f; \kleislirel{\relone});\mu
     \subseteq \relator\reltwo; \mu 
     = \kleislirel{\reltwo}.
 \end{align*}
\end{proof}

\subsubsection*{Proposition~\ref{prop:subject-expansion-eta}.}Let $(\monad, \unit, \mu)$ be weakly cartesian. For any monadic relation  
    $\relone \subseteq A \times \monad(B)$, we have 
    $\eta; \kleislirel{\relone} \subseteq \relone$. 
    In particular, $\eta; \kleislirel{\trel} \subseteq \trel$.
\begin{proof}
    $\eta; \kleislirel{\relone} = \eta; \relator\relone; \mu = \relone; \eta; \mu = \relone$.
\end{proof}

    \subsubsection*{Proposition~\ref{prop:subject-reduction-lifting}.}
    Let $\monad$ be WP. Then:
    $\relone^{\circ}; f \subseteq {\kleislirel{\relone}}^{\circ} \implies 
    {\kleislirel{\relone}}^{\circ}; \kleisli{f} \subseteq {\kleislirel{\relone}}^{\circ}$. 
\begin{proof}
    By duality, it is sufficient to prove 
    $f^\circ; \relone \subseteq \kleislirel{\relone} \implies 
     {f}^{\dagger\circ};\kleislirel{\relone} \subseteq \kleislirel{\relone}$. 
     We have:
     \begin{align*}
          {f}^{\dagger\circ};\kleislirel{\relone} 
          &= (Tf; \mu)^\circ;\Phi R; \mu 
          \\
          &= \mu^\circ; \Phi f^\circ; \Phi R; \mu 
          \\
        &= \mu^\circ; \Phi(f^\circ; R); \mu 
        \\
        &\subseteq \mu^\circ; \Phi(\kleislirel{R}); \mu 
        \tag{by hypothesis $^\circ; \relone \subseteq \kleislirel{\relone}$}
        \\
         &= \mu^\circ; \Phi(\Phi R; \mu); \mu 
         \\
         &= \mu^\circ; \Phi\Phi R; T\mu; \mu
         \\
         &= \mu^\circ; \Phi\Phi R; \mu; \mu
         \tag{Definition of monad}
         \\
         &\subseteq \mu^\circ; \mu; \Phi R; \mu 
         \tag{Definition~\ref{def:relator}}
         \\
         &\subseteq \Phi R; \mu 
         \tag{$\mu$ is a function, hence $\mu^\circ;\mu \subseteq I$}
         \\
         &= \kleislirel{R}.
     \end{align*}
\end{proof}
\begin{lemma}[Weakening Is Admissible]\label{l:weak}
	If $\pi\pof\tjudg{\tye,\var:I}{\tm}{\gty}$, then there exists $\pi'$ such that $\pi'\pof\tjudg{\tye,\tyetwo,\var:I\cup J}{\tm}{\gty}$.
\end{lemma}
\begin{proof}
	Trivial by induction on the structure of $\pi$, since axioms admit weakening. \qed
\end{proof}

\subsection{Subject Reduction}

We use in this appendix the notation $\metalambda\!I_i.M_i$ instead of the notation $\{I_i\Rightarrow\}_{1\leq i\leq n}$.

\begin{lemma}[Substitution]\label{lem:substitution}
	If $\tyd\pof\tjudg{\tye,\var:\intty}{\tm}{\gty}$ and $\tydtwo\pof\tjudg{\tye}{\val}{\intty}$, then there exists a derivation $\tydthree$ such that $\tydthree\pof\tjudg{\tye}{\tm\isub{\var}{\val}}{\gty}$.
\end{lemma}

\begin{proof}
	By induction on the structure of $\pi$.
	
	\begin{itemize}
		\item Case $\rax$. We distinguish two sub-cases: 
		\begin{enumerate}
			\item $\tm = \var$: then $\tm\isub{\var}{\val}= \var\isub{\var}{\val}=\val$ and $\gty = \intty$. The derivation $\tydthree$ is $\tydtwo$ itself.
			
			\item  $ \tm = \vartwo \neq \var $ then $\tm\isub{\var}{\val}= \vartwo\isub{\var}{\val}=\vartwo$, then $\pi''$ is $\pi$. 
		\end{enumerate}
		
		\item Case $ \rabs $. This means that $\tm=\la\vartwo\tmtwo$ and $\gty =\arr\inttytwo\montytwo$ , then $\tm\isub{\var}{\val}=\la\vartwo\tmtwo\isub{\var}{\val}$.
		The derivation $\pi$ has the following shape:
		
		\[
		\begin{array}{c@{\hspace{1cm}}c@{\hspace{1cm}}c}
			\vdots		\\
			\infer[\rabs]{\tjudg{\tye,\var:\intty}{\la\vartwo\tmtwo}{\arr\inttytwo\montytwo}} 
			{\tjudg{\tye,\var:\intty,\vartwo:\inttytwo}{\tmtwo}{\montytwo}}
		\end{array}
		\]
		
		By induction hypothesis there exists a derivation $\pi'''$ such that $\pi'''\pof \tjudg{\tye,y:J}{\tmtwo\isub{\var}{\val}}{\montytwo} $. Applying the $\rabs$ rule one obtains $\pi''\pof {\tjudg{\tye}{\la\vartwo\tmtwo \isub{\var}{\val}}{\arr\inttytwo\montytwo}} $.

		\item Case $ \rapp $.  This means that $\tm = \valtwo\tmtwo$  and $\gty =\montytwo\bind{(\metalambda\!\intty_i.\monty_i)}$, $1\leq i\leq n$ and where $ \supp{\montytwo}\subseteq\{\intty_1,\ldots,\intty_n\} $, then $\tm\isub{\var}{\val}=\valtwo\isub{\var}{\val}\tmtwo\isub{\var}{\val}$. 
		The derivation $\tyd$ has the following shape:
		
		\[
		\begin{array}{c@{\hspace{1cm}}c@{\hspace{1cm}}c}
			\vdots		\\		
			\infer[\rapp]{\tjudg{\tye,\var:\intty}{\valtwo\tmtwo}{\montytwo\bind{(\metalambda\!\intty_i.\monty_i)}}}
			{ \left[ \tjudg{\tye,\var:\intty}{\valtwo}{\arr{\intty_i}{\monty_i}}\right]_{1\leq i\leq n}
				&  \tjudg{\tye,\var:\intty}{\tmtwo}{\montytwo} & \supp{\montytwo}\subseteq\{\intty_1,\ldots,\intty_n\}}
		\end{array}
		\]
		By induction hypothesis there exist $n$ derivations $\tyd_i$ such that, for each $i$, $\tyd_i\pof \tjudg{\tye}{\valtwo\isub{\var}{\val}}{\arr{\intty_i}{\monty_i}} $, and a derivation $\bar{\tyd}\pof \tjudg{\tye}{\tmtwo\isub{\var}{\val}}{\montytwo} $. By applying the $\rapp$ rule we conclude this case. 
		
		\item Case $ \rop $. This means that $\tm=\op{\tm_1\ldots\tm_n}$ and $\gty =\gefop{\monty_1,\ldots,\monty_n}$, then $\tm\isub{\var}{\val}=\op{\tm_1\isub{\var}{\val}\ldots\tm_n\isub{\var}{\val}}$.
		The derivation of $\tyd$ has the following shape
		
		\[
		\begin{array}{c@{\hspace{1cm}}c@{\hspace{1cm}}c}
			\vdots		\\			
			\infer[\rop]{\tjudg{\tye,\var:\intty}{\op{\tm_1,\ldots,\tm_n}}{\gefop{\monty_1,\ldots,\monty_n}}} 
			{\left[
				\tjudg{\tye,\var:\intty}{\tm_i}{\monty_i}\right]_{1\leq i\leq n}}
		\end{array}
		\]
		The thesis follows by i.h., since for each $1\leq i\leq n$ we have derivations $\pi_i\pof \tjudg{\tye}{\tm_i\isub{\var}{\val}}{\monty_i}$, which give a derivation $\pi''\pof \tjudg{\tye}{\op{\tm_1\isub{\var}{\val},\ldots,\tm_n\isub{\var}{\val}}}{\gefop{\monty_1,\ldots,\monty_n}}$.
		
		\item Case $ \runit $.
		Then $\tyd$ has the following shape:
		\[
		\begin{array}{c@{\hspace{1cm}}c@{\hspace{1cm}}c}
			\vdots		\\		
			\infer[\runit]{\tjudg{\tye,\var:\intty}{\valtwo}{\unit(\inttytwo)}}{\tjudg{\tye,\var:\intty}{\valtwo}{\inttytwo}}
		\end{array}
		\]
		with $\tm=\valtwo$ and $\gty = \unit(\inttytwo)$. The thesis follows by i.h. and a similar reasoning as the previous ones. 
		\item Case $ \rint $.
		Then $\tyd$ has the following shape:
		\[
		\begin{array}{c@{\hspace{1cm}}c@{\hspace{1cm}}c}
			\vdots		\\		
			\infer[\rint]{ \tjudg{\tye,\var:\intty}{\valtwo}{\int{\linty_i}_{i\in 
						F}} }{\left[ \tjudg{\tye,\var:\intty}{\valtwo}{\linty_i}\right]_{i\in F}}
		\end{array}
		\]
		with $\tm=\valtwo$ and $\gty = \int{\linty_i}_{i\in F}$. By i.h. there exist $|F|$ derivations $\tyd_i\pof \tjudg{\tye}{\valtwo\isub{\var}{\val}}{\linty_i} $ and the thesis follow by applying $\rint$ from these derivations.
	\end{itemize} \qed
\end{proof}

\subsubsection*{Notation.} For the sake of proofs readability, it is convenient to give an operational description of 
the relation $\kleisli{\vdash}$. First, we recall that any finitely supported monadic element $m \in T(X)$ can be 
represented via a generic effect, in the sense that there exists a generic effect $g$ and elements 
$x_1, \hdots, x_n\in X$ such that $m = g(\eta(x_1), \hdots, \eta(x_n))$ \cite{PlotkinP02}.
Accordingly, we can rephrase the definition of the Barr extension of a relation 
in terms of generic effects
and observe that 
$\kleisli{\vdash} e: g(M_1, \hdots, M_n)$ iff there exists terms $t_1, \hdots, t_n$ such that 
$e = g(\eta(t_1), \hdots, \eta(t_n))$ with 
$\vdash t_i: M_i$, for any $i$. Please notice that the equality $e = g(\eta(t_1), \hdots, \eta(t_n))$ 
is \emph{semantic} equality in $T(\mathbb{C})$, and not a syntactic characterization. 

Indexing the typing relation with typing environment, we come up with 
the following rule for the monadic typing relation $\kleisli{\vdash}$:
\[
\infer[\rextg]{\Gamma \kleisli{\vdash} g(\eta(t_1), \hdots, \eta(t_n)): g(M_1, \hdots, M_n)}
{\Gamma \vdash t_1: M_1 & \cdots & \Gamma \vdash t_n: M_n}
\]
Notice that looking at $\eta$ in terms of generic effects, we also obtain 
the following rule:
\[
\infer[\rextunit]{\Gamma \kleisli{\vdash} \eta(t): M}{\Gamma \vdash t: M}
\]



\subsubsection*{Theorem~\ref{prop:sub-red} (Subject Reduction).} Let $T$ be WP. Then:\begin{varenumerate}
	\item  Let $\tm$ be a closed $\lambda$-term. If $\tjudg{}{\tm}{\monty}$ and $\tm\,\tosem\,\etm$, then 
	$\ltjudg{}{\etm}{\monty}$.
	\item Let $e$ be a monadic closed $\lambda$-term. If $\kleislirel{\vdash} e: M$ 
	and $e \mathrel{\kleisli{\tosem}} e'$, then $\kleisli{\vdash} e' : M$.
\end{varenumerate}
\begin{proof}
	The proof of point 1 proceeds by induction on the structure of evaluation contexts $\evctx$.
	\begin{itemize}
		\item $\evctx=\ctxhole$. There are three sub-cases.
		\begin{itemize}
			\item $(\la\var\tmthree)\val\mapstob\unit(\tm\isub\var\val)$.
			By hypothesis we have:
			\[\infer{\tjudg{\tye}{(\la\var\tmthree)\val}{\montytwo\bind{(\metalambda\!\intty_i.\monty_i)}}}
			{ \left[ \tjudg{\tye}{\la\var\tmthree}{\arr{\intty_i}{\monty_i}}\right]_{1\leq i\leq n}
				&  \tjudg{\tye}{\val}{\montytwo} & \supp{\montytwo}\subseteq\{\intty_1,\ldots,\intty_n\}}\]
			We note that $\val$ can be typed with a monadic type just via the rule $\runit$, as $\montytwo=\unit(\intty)$ obtaining the derivation
			\[
			\infer{\tjudg{\tye}{\val}{\intty}}{\vdots}
			\] This way $n=1$ and $\montytwo\bind{(\metalambda\!\intty_i.\monty_i)}=\unit(\intty)\bind(\metalambda\!\intty.\monty)=\monty$. Then, the left premise can only be obtained through an application of the rule $\rabs$. Going backward we have the derivation: \[\tjudg{\tye,\var:\intty}{\tmthree}{\monty}\]
			By applying the substitution lemma, we obtain 
			\[\infer{\tjudg{\tye}{\tmthree\isub{\var}{\val}}{\monty}}{\vdots}\]
			Then one immediately obtains, via rule $\rextunit$:
			\[{\tye}\kleisli{\vdash}{\unit(\tmthree\isub{\var}{\val})}:{\monty}\]
			
			\item $\op{\tm_1,\ldots,\tm_n}\mapstoe\gefop{\unit(\tm_1),\ldots,\unit(\tm_n)}$. By hypothesis, we have:
			\[
			\infer{\tjudg{\tye}{\op{\tm_1,\ldots,\tm_n}}{\gefop{\monty_1,\ldots,\monty_n}}} 
			{\left[\tjudg{\tye}{\tm_i}{\monty_i}\right]_{1\leq i\leq n}}
			\]
			It is immediate to derive:
			\[\infer[\rextg]{{\tye}\kleisli{\vdash}{\gefop{\unit(\tm_1),\ldots,\unit(\tm_n)}}:{\gefop{\monty_1,\ldots,\monty_n}}} 
			{\left[\tjudg{\tye}{\tm_i}{\monty_i}\right]_{1\leq i\leq n}}\]
			\item $\val\mapstoeta\unit(\val)$.  We have
			\[
			\infer[\rextunit]{\ltjudg{}{\unit(\val)}{\monty}}{
				\infer{\tjudg{}{\val}{\monty}}{\text{hypothesis}}}
			\]
		\end{itemize}
		\item 
		$\evctx=\val\evctx'$. Thus we have $\tm=\val\evctx'\ctxholep{\tmthree} \to 
		\gef{\unit(\val\evctx'\ctxholep{\tmthree'_1})\ldots 
			\unit(\val\evctx'\ctxholep{\tmthree'_n})}=:\etm$, if 
		$\tmthree\to\gef{\unit(\tmthree_1') 
			\ldots\unit(\tmthree_n')}$. By \ih, 
		we have that if $\evctx'\ctxholep{\tmthree} \to 
		\gef{\unit(\evctx'\ctxholep{\tmthree'_1})\ldots 
			\unit(\evctx'\ctxholep{\tmthree'_n}})$ and 
		$\tjudg{}{\evctx'\ctxholep{\tmthree}}{\monty}$, then 
		$\ltjudg{}{\gef{\unit(\evctx'\ctxholep{\tmthree'_1})\ldots 
				\unit(\evctx'\ctxholep{\tmthree'_n})}}{\monty}$. By hypothesis we have
		\[
		\infer{\tjudg{}{\val\evctx'\ctxholep{\tmthree}}{\monty \bind{(\metalambda\!\intty_j.\monty_j)}}}
		{ \infer*{\left[\tjudg{}{\val}{\arr{\intty_j}{\monty_j}}\right]_{1\leq j\leq 
					m}}{} 
			& \infer*{\tjudg{}{\evctx'\ctxholep{\tmthree}} 
				{\monty}}{} & \supp{\monty}\subseteq\{\intty_1\ldots\intty_m\}}
		\]
		By \ih we have $\ltjudg{}{\gef{\unit(\evctx'\ctxholep{\tmthree'_1})\ldots 
				\unit(\evctx'\ctxholep{\tmthree'_n})}}{\monty}$. This means that $\ltjudg{}{\gef{\unit(\evctx'\ctxholep{\tmthree'_1})\ldots
				\unit(\evctx'\ctxholep{\tmthree'_n})}}{\gef{\montytwo_1\ldots\montytwo_n}}$ for some $\montytwo_1\ldots \montytwo_n$. Of course we have that $\supp{\montytwo_i}\subseteq\supp{\monty}\subseteq\{\intty_1\ldots\intty_m\}$. Moreover, we have:
		\[
		\infer{\ltjudg{}{\gef{\unit(\evctx'\ctxholep{\tmthree'_1})\ldots 
					(\evctx'\ctxholep{\tmthree'_n})}}{\gef{\montytwo_1\ldots\montytwo_n}}}{\left[ 
			\infer{\tjudg{}{\evctx'\ctxholep{\tmthree'_i}}{\montytwo_i}}{\vdots}\right]_{1\leq
				i\leq n}}
		\]
		Thus, we have
		\[
		\infer{\ltjudg{}{\gef{\unit(\val\evctx'\ctxholep{\tmthree'_1})\ldots 
					\unit(\val\evctx'\ctxholep{\tmthree'_n})}}{\gef{\montytwo_1 \bind{(\metalambda\!\intty_j.\monty_j)}
					\ldots \montytwo_n \bind{(\metalambda\!\intty_j.\monty_j)}}}}
		{\left[ 
			\infer{\tjudg{}{\val\evctx'\ctxholep{\tmthree'_i}}{\montytwo_i \bind{(\metalambda\!\intty_j.\monty_j)}}} 
			{\left[\tjudg{}{\val}{\arr{\intty_j}{\monty_j}}\right]_{1\leq j\leq m} &&  
				\tjudg{}{\evctx'\ctxholep{\tmthree'_i}}{\montytwo_i} && \supp{\montytwo_i}\subseteq\{\intty_1\ldots\intty_m\}}  \right]_{1\leq
				i \leq n} }
		\]
		Since $\gef{\montytwo_1 \bind{(\metalambda\!\intty_j.\monty_j)}
			\ldots \montytwo_n \bind{(\metalambda\!\intty_j.\monty_j)}} = \gef{\montytwo_1\ldots\montytwo_n} \bind {(\metalambda\!\intty_j.\monty_j)} = \monty \bind {(\metalambda\!\intty_j.\monty_j)}$, we conclude.
	\end{itemize}
Point 2 is then obtained by applying Proposition~\ref{prop:subject-reduction-lifting}. \qed
\end{proof}

\subsection{Logical Relations}

\subsubsection{Lemma~\ref{l:semlr}.}
$\trelsem t: M$ implies $\exists e\text{ such that }\sem t = e$ and $\obsfun(e) = \obsfun(M)$.
\begin{proof}
 If $\trelsem t: M$, then  $\widehat{\trelsem} \sem t : M$. We set $e\defeq\sem t$.
 We show that $\widehat{\trelsem} \sem t : M$ implies $\obs{e} = \obsfun(M)$, 
 i.e. $\widehat{\trelsem}; \obsfun \subseteq \obsfun; \idrel$. 
 This is actually an instance of a more general result, namely that 
 for any $\relone: A \to B$, we have 
 $\bextt\relone; T(!) \subseteq T(!); \idrel$. Indeed, 
 recalling that $\idrel_{TA} = \bextt{\idrel_A}$, we have:
 \begin{align*}
         \bextt\relone; T(!) \subseteq T(!); \idrel 
         &\iff 
         \bextt\relone; T(!) \subseteq T(!); \bextt\idrel
         \\
       &\impliedby 
         \relone; ! \subseteq\ !; \idrel
     \end{align*}
 and the latter obviously holds. \qed
\end{proof}

As usual, we extend $\trelsem$ to open terms thus:
$x_1: I_1, \hdots, x_n: I_n \trelsem t: G$ iff for all $v_1, \hdots, v_n$ 
such that $\trelsem v_i: I_i$, we have $\trelsem t\{x_1/v_1, \hdots, x_n/v_n\} : G$.
To improve readability, we sometimes denote substitutions as $\gamma, \gamma', \hdots$. 
Notice that if a substitution is \emph{closed}, meaning that in a substitution  
$\{x_1/v_1, \hdots, v_n/x_n\}$ all values $v_i$ are closed, then sequential and simultaneous substitution 
coincide. In particular, $t\{x_1/v_1, \hdots, x_n/v_n\}\{x_{n+1}/v_{n+1}\} =  t\{x_1/v_1, \hdots, x_n/v_n,x_{n+1}/v_{n+1}\}$. 
In the following, all substitutions considered are closed.

\begin{lemma}\label{l:saturation}
	If $\tm\mapsto\etm$, then $\trelsem \tm:\monty$ if and only if $\hat{\trelsem}\,\tm:\monty$.
\end{lemma}
\begin{proof}
	Trivial by definition of $\trelsem$, since the function $\sem -$ is invariant by reduction/expansion.
\end{proof}

\subsubsection*{Proposition~\ref{prop:soundness-lr} (Soundness).}
	${\trel} = {\trelsem}$.
\begin{proof}
	The inclusion ${\trelsem} \subseteq {\trel}$ holds by definition. We prove ${\trel} \subseteq {\trelsem}$ by induction on the definition of $\trel$. In particular, we prove the stronger statement: $\tjudg{\tye}{\tm}{\gty}$ implies $\tye\trelsem\tm:\gty$.
	\begin{itemize}
		\item Case $\rax$.
		\[
		\infer[\rax]{\tjudg{\tye,\var:\intty}\var\linty}{\linty\in\intty}
		\]
		Wlog we consider $\tye$ empty. We have to prove that $\trelsem v: I$ implies $\trelsem \var\isub\var\val : A$. By definition, $\trelsem v: I$ implies $\trelsem_\Tys\val:\linty$ for each $\linty\in\intty$.
		\item Case $\rabs$.
		\[\infer[\rabs]{\tjudg{\tye}{\la\var\tm}{\arr\intty\monty}} 
		{\tjudg{\tye,\var:\intty}{\tm}{\monty}}\]
		By \ih we have that
		$	\tye,\var:\intty\sjudg{\tm}{\monty}$. We have to prove that for each $\valtwo$ such that $\sjudg{\valtwo}{\intty}$, we have $\tye\sjudg{(\la\var\tm)\valtwo}{\monty}$. Since $(\la\var\tm)\valtwo\mapsto\unit(\tm\isub\var\valtwo)$, we conclude by applying the \ih, and going backward through Lemma~\ref{l:saturation}.
		
		\item Case $\rop$.
		\[
		\infer[\rop]{\tjudg{\tye}{\op{\tm_1,\ldots,\tm_n}}{\gefop{\monty_1,\ldots,\monty_n}}} 
		{\left[
			\tjudg{\tye}{\tm_i}{\monty_i}\right]_{1\leq i\leq n}}
		\]
		By \ih we have that $\tye\sjudg{\tm_i}{\monty_i}$ for each $i$. Then we have  $\tye\,\lsjudg{\gefop{\unit(\tm_1)\ldots,(\tm_n)}}{\gefop{\monty_1,\ldots,\monty_n}}$ by rule $\rextg$, which gives the thesis applying Lemma~\ref{l:saturation}.
		\item Case $\rapp$.
		\[\infer[\rapp]{\tjudg{\tye}{\val\tm}{\montytwo\bind{(\intty_i\mapsto\monty_i)}}}
		{ \left[ \tjudg{\tye}{\val}{\arr{\intty_i}{\monty_i}}\right]_{1\leq i\leq n}
			&  \tjudg{\tye}{\tm}{\montytwo} & \supp{\montytwo}\subseteq\{\intty_1,\ldots,\intty_n\}}\]
		By induction hypothesis, we have 
            $\Gamma \models I_i \to M_i$, for any $i \leq n$ and
            $\Gamma \models t: N$. Let $\gamma$ a substitution for $\Gamma$, so that we have 
            \begin{align}
                \forall i \leq n.\ \models v\gamma: I_i \to M_i 
                \label{soundness:aux-1}
                \tag{$\dag$}
            \end{align}
             and
             $\models t\gamma: N$, and thus 
             \begin{align}
                \mathrel{\hat{\models}} \sem{t}: N
                \label{soundness:aux-2}
                \tag{$\ddag$}
            \end{align}
            To prove the thesis, it is sufficient to show
            $\mathrel{\hat{\models}} \sem{v\gamma t\gamma}: N \bind f$, where 
            $f: \{I_1, \hdots, I_n\} \to T(\mathbb{I})$ maps 
            each $I_i$ to $M_i$. 
            Since $v\gamma$ is a closed value (this follows from $v$ being a value 
            and $\gamma$ substituting closed values to variables), we have 
            $v\gamma = \lambda x.s$, for some term $s$. Consequently, we have
            Since $\sem{v\gamma t\gamma} = \sem{t\gamma} \bind \sem{s[-/x]}$, so that 
            we can reduce the thesis to the proof obligation 
            $$\mathrel{\hat{\models}} \sem{t\gamma} \bind \sem{s[-/x]}: N \bind f.$$
            We are now in the position to use 
            Theorem~\ref{thm:barr} (actually the weaker Proposition~\ref{prop:right-barr-extension-is-a-lax-extension} suffices), this way 
            obtaining the proof obligations:
            \begin{align*}
                &\mathrel{\hat{\models}} \sem{t\gamma}: N 
                \\
                \forall w,I_j.\ \models w: I_j &\implies 
                \mathrel{\hat{\models}} \sem{s[w/x]}: f(I_j).
            \end{align*}
		The former is exactly \eqref{soundness:aux-1}, whereas for the latter we reason thus.
            Assume $\models w: I_j$. Then, by \eqref{soundness:aux-1}, we have 
            $\models v\gamma w: M_j$ (recall that $f(I_j) = M_j$), which gives 
            $\mathrel{\hat{\models}} \sem{v\gamma w}: M_j$. Now, since 
            $v\gamma = \lambda x.s$, we have 
            $\sem{v\gamma w} = \sem{(\lambda x.s)w} = \sem{s[w/x]}$. We are done.
		\item Case $\rint$.
		\[\infer[\rint]{ \tjudg{\tye}{\val}{\int{\linty_i}_{i\in 
					F}} }{\left[ \tjudg{\tye}{\val}{\linty_i}\right]_{i\in F}}\]
		By \ih we have that  $\forall i\in F,\ \tye\trelsem_{\mathbb{A}} v: A_i$. Then we conclude by definition.
		\item Case $\runit$.
		\[
		\infer[\runit]{\tjudg{\tye}{\val}{\unit(\intty)}}{\tjudg{\tye}{\val}{\intty}}
		\]
		By \ih we have that $\tye\trelsem_{\mathbb{I}} v: \intty$. We conclude, by applying the definition of $\trelsem$, since $\tye\,\lsjudg{\unit(\val)}{\unit(\intty)}$ by $\rextunit$.
	\end{itemize}
\end{proof}

\subsection{Subject Expansion}

\begin{lemma}[Anti-Substitution]\label{l:antisub}
	If $	\pi\pof\tjudg{\tye}{\tm\isub{\var}{\val}}{G}$, then there exists an intersection type $\intty$ and derivations $\pi'$ and $\pi''$ such that $\pi'\pof \tjudg{\tye}{\val}{\intty}$ and $\pi''\pof \tjudg{\tye,\var:\intty}{\tm}{G}$.
\end{lemma}
\begin{proof}
	
	By induction on the structure of $\pi$.
%
\begin{itemize}
	\item Case $\rax$, $\tm=\var$. Then $\pi'=\pi$, and  $\pi''\pof\tjudg{\tye,\var:\intty}{\var}{\intty}$ by $\rax$ axiom.
	
	\item Case $\rax$, $\tm=\vartwo\neq\var$. We have that $\tm\isub{\var}{\val}=\vartwo\isub{\var}{\val}=\vartwo$, so that we obtain the thesis by taking $\pi''=\pi$, and $\pi'\pof\tjudg{\tye}{\val}{\emi}$.
	
	\item Case $\rabs$, $\tm=\la\vartwo\tmtwo$. then $\tm\isub{\var}{\val}=\la\vartwo\tmtwo\isub{\var}{\val}$; in particular, since free variables in $\val$ cannot be caught in $\tmtwo\isub{\var}{\val}$ by the binding $\lambda \vartwo$, we assume $\vartwo\not\in FV(\val)$.
	$\gty= \arr\intty\monty$. 
	So, from 
	
	\[
	\begin{array}{c@{\hspace{1cm}}c@{\hspace{1cm}}c}
		\vdots		\\
		\infer[\rabs]{\tjudg{\tye}{\la\vartwo\tm\isub{\var}{\val}}{\arr\intty\monty}} 
		{\tjudg{\tye,\vartwo:\intty}{\tm\isub{\var}{\val}}{\monty}}
	\end{array}
	\]
	by i.h. there exists $\inttytwo$ such that  $ \tjudg{\tye,\vartwo:\intty,\var:\inttytwo}{\tm}{\monty} $ and $ \pi'\pof\tjudg{\tye,\vartwo:\intty}{\val}{\inttytwo} $. From this we derive $\pi''$ as:
	
	\[
	\begin{array}{c@{\hspace{1cm}}c@{\hspace{1cm}}c}
		\vdots		\\
		\infer[\rabs]{\tjudg{\tye,\var:\inttytwo}{\la\vartwo\tm}{\arr\intty\monty}} 
		{\tjudg{\tye,\vartwo:\intty,\var:\inttytwo}{\tm}{\monty}}
	\end{array}
	\]
	
	\item Case $\rapp$, $\tm = \valtwo\tmtwo$. Then $\tm\isub{\var}{\val}=\valtwo\isub{\var}{\val}\tmtwo\isub{\var}{\val}$ and $\gty= \montytwo\bind{(\metalambda\!\intty_i.\monty_i)}$
	
	\[
	\begin{array}{c@{\hspace{1cm}}c@{\hspace{1cm}}c}
		{\infer[\rapp]{\tjudg{\tye}{\valtwo\tmtwo\isub{\var}{\val}}{\montytwo\bind{(\metalambda\!\intty_i\monty_i)}}}
			{ 
				\left[ \tjudg{\tye}{\valtwo\isub{\var}{\val}}{\int{\arr{\intty_i}{\monty_i}}}\right]_{1\leq i\leq n}
				&  \tjudg{\tye}{\tmtwo\isub{\var}{\val}}{\montytwo} & \supp{\montytwo}\subseteq\{\intty_1,\ldots,\intty_n\}}}
	\end{array}
	\]
	
	by induction hypothesis there exist $\inttytwo_i$ and $\bar{\inttytwo}$ such that for each $1\leq i\leq n$, there are $\tjudg{\tye,\var:\inttytwo_i}{\valtwo}{\arr{\intty_i}{\monty_i}}$ and $\tjudg{\tye}{\val}{\inttytwo_i}$, $\tjudg{\tye, x:\bar\inttytwo}{\tmtwo}{\montytwo}$ and  $\tjudg{\tye}{\val}{\bar{\inttytwo}}$. We set $\bar{\bar{\inttytwo}}\defeq \bigcup_i J_i \cup \bar{\inttytwo}$. Then, by Lemma~\ref{l:weak}, it is immediate to obtain $\pi'\pof\tjudg{\tye}{\val}{\bar{\bar{\inttytwo}}}$ and $\pi''\pof\tjudg{\tye,\var:\bar{\bar{\inttytwo}}}{\valtwo\tmtwo}{\montytwo\bind{(\metalambda\!\intty_i.\monty_i)}}$.
	
	\item Case $\rop$, $\tm = \op{\tm_1,\ldots,\tm_n}$. Then $\tm\isub{\var}{\val} = \op{\tm_1\isub{\var}{\val},\ldots,\tm_n\isub{\var}{\val}}$ and $\gty= \gefop{\monty_1,\ldots,\monty_n}$.
	
	\[
	\begin{array}{c@{\hspace{1cm}}c@{\hspace{1cm}}c}
		\infer[\rop]{\tjudg{\tye}{\op{\tm_1\isub{\var}{\val},\ldots,\tm_n\isub{\var}{\val}}}{\gefop{\monty_1,\ldots,\monty_n}}} 
		{\left[
			\tjudg{\tye}{\tm_i\isub{\var}{\val}}{\monty_i}\right]_{1\leq i\leq m} }
	\end{array}
	\]
	
	By induction hypothesis there exist $\inttytwo_i$ such that $\tjudg{\tye, \var: \inttytwo_i}{\tm_i}{\monty_i}$ and $ \tjudg{\tye}{\val}{\inttytwo_i}$, for each $i$. 
	We then easily obtain (thanks to Lemma~\ref{l:weak}) $\pi''\pof\tjudg{\tye, \var: \bar{\inttytwo}}{\op{\tm_1, \ldots, \tm_n}}{\gefop{\monty_1,\ldots,\monty_n}}$ and $\pi'\pof\tjudg{\tye}{\val}{\bar{\inttytwo}}$, where $\bar{\inttytwo}\defeq\bigcup_i\inttytwo_i$.
	\item The remaining cases for rules $\runit$ and $\rint$ are handled similarly.
\end{itemize}	\qed
\end{proof}

\subsubsection{Proposition~\ref{prop:subject-expansion}. (Subject Expansion)}
Let $T$ be WC. Then:\begin{enumerate}
    \item 	If $\tm \mathrel{\tosem} e$ and $\kleisli{\vdash} e : M$, 
        then $\vdash \tm: M$.
	\item If $e \mathrel{\kleisli{\tosem}} e'$ and $\kleisli{\vdash}e' : M$, 
	    then $\kleisli{\vdash} e : M$.
\end{enumerate}
\begin{proof}
	We prove the first point by induction on the structure of evaluation context $\evctx$. 
 \begin{itemize}
 \item $\evctx=\ctxhole$. There are three sub-cases.
 \begin{itemize}
     \item Case $(\la\var\tmthree)\val\mapstob\unit(\tm\isub\var\val)$. The proof is equivalent to the subject reduction one, read backwards, and using  Lemma~\ref{l:antisub} (anti substitution) instead of Lemma~\ref{lem:substitution} (substitution).
     \item Case $\op{\tm_1,\ldots,\tm_n}\mapstoe\gefop{\unit(\tm_1),\ldots,\unit(\tm_n)}$. The proof is equivalent to the subject reduction one, read backwards, and using the fact that operations do not cancel arguments.
     \item Case $\val\mapstoeta\unit(\val)$. By hypothesis we have:
     \[
    \infer{\ltjudg{}{\unit(\val)}{\monty}}{\vdots}
     \]
     Proposition~\ref{prop:subject-expansion-eta} guarantees that $M=\unit(I)$ for some $I$. From this fact, we obtain $\tjudg{}{\val}{\unit(I)}$.
 \end{itemize}
 \item 
	$\evctx=\val\evctx'$. Thus we have $\tm=\val\evctx'\ctxholep{\tmthree} \to 
	\gef{\unit(\val\evctx'\ctxholep{\tmthree'_1})\ldots 
		\unit(\val\evctx'\ctxholep{\tmthree'_n})}=\etm$, if 
	$\tmthree\to\gef{\unit(\tmthree_1') 
		\ldots\unit(\tmthree_n')}$. By \ih, 
	we have that if $\evctx'\ctxholep{\tmthree} \to 
	\gef{\unit(\evctx'\ctxholep{\tmthree'_1})\ldots 
		\unit(\evctx'\ctxholep{\tmthree'_n})}$ and 
	$\tjudg{}{\gef{\unit(\evctx'\ctxholep{\tmthree'_1})\ldots 
			\unit(\evctx'\ctxholep{\tmthree'_n})}}{\monty}$, then 
	$\tjudg{}{\evctx'\ctxholep{\tmthree}}{\monty}$. By hypothesis
	we have
	\[
	\infer{\ltjudg{}{\gef{\unit(\val\evctx'\ctxholep{\tmthree'_1})\ldots 
	\unit(\val\evctx'\ctxholep{\tmthree'_n})}}{\gef{\monty_1\ldots\monty_n}}}
	{\left[  
	\infer{\tjudg{}{\val\evctx'\ctxholep{\tmthree_i}}{\monty_i\defeq 
	\montytwo_i \bind {(\intty_{ij}\mapsto\monty_{ij})}}} 
	{\infer*{\left[\tjudg{}{\val}{\arr{\intty_{ij}}{\monty_{ij}}}\right]_{1\leq
	 j\leq m_i}}{} & \infer*{\tjudg{}{\evctx'\ctxholep{\tmthree_i}} 
			{\montytwo_i}}{} && \supp{\montytwo_i}\subseteq\{\intty_{i1}\ldots\intty_{im_i}\}}\right]_{1\leq i\leq n} }
	\]
%
	Thus we can type $\gef{\unit(\evctx'\ctxholep{\tmthree'_1})\ldots 
		\unit(\evctx'\ctxholep{\tmthree'_n})}$ as follows:
	\[
	\infer{\tjudg{}{\gef{\unit(\evctx'\ctxholep{\tmthree'_1})\ldots 
				\unit(\evctx'\ctxholep{\tmthree'_n})}}{\gef{\montytwo_1
				 \ldots \montytwo_n}}}{
			\left[ \tjudg{}{\evctx'\ctxholep{\tmthree_i}} 
			{\montytwo_i}\right]_{1\leq i \leq n} }
	\]
	
	By \ih, we have that 
	$\tjudg{}{\evctx'\ctxholep{\tmthree}}{\gef{\montytwo_1
			\ldots \montytwo_n}}$.

	Now, we are able to type $\val\evctx'\ctxholep{\tmthree}$.
	\[
	\infer{\tjudg{}{\val\evctx'\ctxholep{\tmthree}}{\gef{\montytwo_1
				\ldots \montytwo_n} \bind {(\intty_{ij}\mapsto\monty_{ij})}}}{
		\left[\tjudg{}{\val}{\arr{\intty_{ij}}{\monty_{ij}}}\right]_{1\leq i\leq 
		n,1\leq j\leq m_i} & 
		\tjudg{}{\evctx'\ctxholep{\tmthree}}{\gef{\montytwo_1
				\ldots \montytwo_n}} && \supp{\gef{\montytwo_1
				\ldots \montytwo_n}}\subseteq\{\intty_{11}.. 
			\intty_{1m_1}.. \intty_{n1}.. \intty_{nm_n}\}}
	\]
	Since $\monty_i\defeq 
	\montytwo_i \bind {(\intty_{ij}\mapsto\monty_{ij})}$, we have $\gef{\montytwo_1
		\ldots \montytwo_n} \bind {(\intty_{ij}\mapsto\monty_{ij})} = \gef{\montytwo_1 \bind {(\intty_{ij}\mapsto\monty_{ij})}
		\ldots \montytwo_n \bind {(\intty_{ij}\mapsto\monty_{ij})}}  = \gef{\monty_1
		\ldots \monty_n}$.\end{itemize}
  The second point is obtained directly applying Proposition~\ref{prop:subject-expansion-lifting}
\end{proof}

\subsubsection{Theorem~\ref{thm:completeness}. (Completeness)}
	If $\sem\tm=\etm$, then there exists a monadic type  $\monty$ such that $\tjudg{}{\tm}{\monty}$.
\begin{proof}
    We prove the following stronger statement: 
    if $\sem\etmtwo =\etm$ (meaning that 
    $e' {\kleislirel{\tosem}}^n\etm$ and $\supp{e} \subseteq \Vals$), 
    then there exists $\monty$ such that 
    $\kleislirel{\vdash} \etmtwo: \monty$. 
	We proceed by induction on $n$.
	\begin{itemize}
		\item 
		If $n=0$, since operations are non-erasing then $\etm=\etmtwo=\gef{\unit(\val_1),\ldots,\unit(\val_m)}$ for some $v_1,\ldots,v_m$. Then, it is immediate to type $\etm$:
		\[
            \dfrac{\left[
            \dfrac{
				{ \tjudg{}{\val_i}{\emi}}}{\tjudg{}{\val_i}{\unit(\emi)}}
            \right]_{1\leq i\leq m}}
            {\ltjudg{}{\gef{\unit(\val_1),\ldots,\unit(\val_m)}}{\gef{\unit(\emi),\ldots,\unit(\emi)}}}
		\]
		\item If $n>0$, then $\etmtwo\kleislirel{\tosem}\etmthree\tosem^{\dagger n-1}\etm$. By \ih there is a type $\monty$ such that $\tjudg{}{\etmthree}{\monty}$. Then by the subject expansion, we conclude that $\tjudg{}{\etmtwo}{\monty}$.
	\end{itemize}
\end{proof}
\section{Omitted Proofs of Section~\ref{sec:infchar}}
Proofs of Proposition~\ref{prop:subject-reduction-infinitary}, Proposition~\ref{prop:soundnessI}
, and Proposition~\ref{prop:subject-expansion-inf} are exactly as in the finitary case.

\subsubsection{Lemma~\ref{l:semlrinf}.}$\trelsem t: M$ implies $\obs{\sem{t}} \cpogeq \obsfun(M)$.
\begin{proof}
It is sufficient to prove the following more general result: 
for any $\relone: A \to B$, we have 
$\bexttinf{\relone}; T(!) \subseteq T(!); {\cpogeq}$. Indeed, 
recalling that $\idrel_{TA} = \bextt{\idrel_A}$, so that 
$\bexttinf{\idrel_{A}} = \bextt{\idrel_A}; {\cpogeq} = \idrel_{TA}; {\cpogeq} =  {\cpogeq}$,
we have:
\begin{align*}
     \bexttinf\relone; T(!) \subseteq T(!);  {\cpogeq}
     &\iff 
     \bexttinf\relone; T(!) \subseteq T(!); \bexttinf{\idrel}
     \\
   &\iff 
     \relone; ! \subseteq !; \idrel
 \end{align*}
and the latter obviously holds. \qed
\end{proof}

\subsubsection{Theorem~\ref{thm:completenessI}. (Approximate Completeness)}
Let $\tm$ be a closed $\l$-term. Then, for each $k\geq 0$, there exist a derivation $\pi_k\pof\tjudg{}\tm{\monty_k}$ such that $\obs{\monty_k}=\obsn{\tm}{k}$.
\begin{proof}
We prove the following stronger statement. Let $\etm$ be a monadic term. Then, for each $k\geq 0$, there exist a derivation $\pi_k\pof\ltjudg{}\etm{\monty_k}$ such that $\obs{\monty_k}=\obsn{\etm}{k}$.
    If $\etm$ is a monadic value $\gef{\unit(\val_1),\ldots,\unit(\val_m)}$, then the derivation below satisfies the claim for each $k$:
    \[
            \dfrac{\left[
            \dfrac{
				{ \tjudg{}{\val_i}{\emi}}}{\tjudg{}{\val_i}{\unit(\emi)}}
            \right]_{1\leq i\leq m}}
            {\ltjudg{}{\gef{\unit(\val_1),\ldots,\unit(\val_m)}}{\gef{\unit(\emi),\ldots,\unit(\emi)}}}
		\]
    If $\etm$ is a monadic term $e=g(\unit(t_1),\ldots,\unit(t_m))$, we proceed by induction on $k$. If $k=0$, then we can type $\etm$ using the $\rbot$ rule as follows:
    \[
            \dfrac{\left[
				\tjudg{}{\tm_i}{\bot}
            \right]_{1\leq i\leq m}}
            {\ltjudg{}{\gef{\unit(\val_1),\ldots,\unit(\val_m)}}{\gef{\bot,\ldots,\bot}}}
		\]
    Otherwise, if $k>0$ we observe that if $\etm\mapsto\etm'$, then $\obsn{\etm'}{k-1}=\obsn{\etm}{k}$. Then, we apply the induction hypothesis to $\etm$, and conclude by subject expansion.\qed
\end{proof}

\subsubsection{Corollary~\ref{c:char-inf}.} Let $\tm$ be a closed $\l$-term. Then $O(\tm)=\obs{\tm}$.
\begin{proof}
Let   $\pi\pof\tjudg{}{\tm}{\monty}$ be a type derivation. Then by Prop.~\ref{prop:soundnessI} we have $\obs{\sem\tm}\cpogeq \obs{M}$. Hence $\obs{\tm}\cpogeq O(\tm)$. By Theorem~\ref{thm:completenessI}, for each $k\geq 0$, there exists a derivation $\pi_k\pof\tjudg{}\tm{\monty_k}$ such that $\obs{\monty_k}\cpogeq\obsn{\tm}{k}$. Hence $O(\tm)\cpogeq \obs{\tm}$, and thus $O(\tm)= \obs{\tm}$.\qed
\end{proof}

\end{document}